\documentclass{amsart}
\usepackage{amssymb}
\usepackage{lmodern}
\usepackage[T1]{fontenc}
\usepackage[utf8]{inputenc}
\usepackage{microtype}
\usepackage[usenames, dvipsnames]{xcolor}
\usepackage[backref=page, colorlinks=true]{hyperref}
\hypersetup{
 linkcolor={red!50!black},
 citecolor={green!50!black},
 urlcolor={blue!80!black}
}
\usepackage{tikz-cd}
\usepackage[margin=2in]{geometry}

\newcommand{\bea}{\begin{eqnarray}}
\newcommand{\eea}{\end{eqnarray}}
\newcommand{\be}{\begin{equation}}
\newcommand{\ee}{\end{equation}}
\newcommand{\ba}{\begin{aligned}}
\newcommand{\ea}{\end{aligned}}

\usetikzlibrary{shapes.geometric,calc}
\usetikzlibrary{shapes,arrows,chains}
\usetikzlibrary{decorations.markings}
\usetikzlibrary{decorations.pathmorphing}
\usetikzlibrary{shapes.multipart}
\tikzset{snake it/.style={decorate, decoration=snake}}
\tikzstyle{GraphNode}=[circle, draw=black, fill=black, inner sep=2pt, minimum size=5pt]
\tikzstyle{GraphEdge}=[black]
\usepackage{subcaption}
\usepackage{setspace}
\usetikzlibrary{calc}
\usetikzlibrary{tikzmark}
\usepackage{comment}
\newcommand\arXiv[1]{\href{http://arxiv.org/abs/#1}{\nolinkurl{arXiv:#1}}}
\newcommand\MRnumber[1]{\href{http://www.ams.org/mathscinet-getitem?mr=#1}{\nolinkurl{MR#1}}}
\newcommand\DOI[1]{\href{http://dx.doi.org/#1}{\nolinkurl{DOI:#1}}}
\newcommand\MAILTO[1]{\href{mailto:#1}{\nolinkurl{#1}}}

\newcounter{mainthm}

\newtheorem{maintheorem}[mainthm]{Theorem}
\newtheorem{dummy}{Dummy}[section]

\newtheorem{lemma}[dummy]{Lemma}
\newtheorem{proposition}[dummy]{Proposition}
\newtheorem{corollary}[dummy]{Corollary}

\newtheorem{theorem}[dummy]{Theorem}
\newtheorem{defn}[dummy]{Definition}

\newtheorem{construction}[dummy]{Construction}
\newtheorem{ansatz}[dummy]{Ansatz}

\theoremstyle{definition}

\newtheorem{rem}{Remark}

\newtheorem{example}{Example}
\newtheorem*{prop*}{Proposition}

\usepackage{dsfont}
\renewcommand\mathbb\mathds

\newcommand\bC{\mathbb C}

\newcommand\Z{\mathbb Z}

\newcommand\cA{\mathcal A}
\newcommand\cB{\mathcal B}
\newcommand\cC{\mathcal C}
\newcommand\cD{\mathcal D}
\newcommand\cE{\mathcal E}
\newcommand\cF{\mathcal F}

\newcommand\cM{\mathcal M}

\newcommand\cT{\mathcal T}

\newcommand\cZ{\mathcal Z}

\newcommand\rC{\mathrm C}

\newcommand\rH{\mathrm H}

\newcommand\SH{\mathrm {SH}}
\newcommand\SO{\mathrm {SO}}

\usepackage[mathscr]{euscript}

\newcommand\fC{\mathfrak C}

\DeclareMathOperator\homology{H}
\renewcommand\H{\homology}

\newcommand{\ko}{\mathit{ko}}

\newcommand{\Sq}{\mathrm{Sq}}

\newcommand{\Spin}{\mathrm{Spin}}

\newcommand{\Mext}{\mathrm{Mext}}

\newcommand{\tsVect}{\mathbf{2sVect}}
\newcommand{\sVect}{\cat{sVect}}

\newcommand{\Witt}{\mathcal{W}itt}
\newcommand{\sBrPic}{\mathscr{B}r\mathscr{P}ic}

\newtheorem{question}[equation]{Question}

\newcommand\longto\longrightarrow
\newcommand\mono\hookrightarrow
\newcommand\epi\twoheadrightarrow

\newcommand\<\langle
\renewcommand\>\rangle
\newcommand\sminus\smallsetminus

\DeclareMathOperator\End{End}

\newcommand\Rep{\cat{Rep}}

\newcommand\cat[1]{\mathbf{#1}}

\title{ Mutual Influence  of Symmetries and Topological Field Theories}

\author{Daniel Teixeira}
 \address{ $^1$ Department of Mathematics, Dalhousie University, Halifax, Nova Scotia }
 
  \author{Matthew Yu}
  \address{ $^2$ Mathematical Institute, University of Oxford, OX2 6GG, Oxford, UK}

\date{}

\begin{document}

\maketitle
 \hspace{1cm}
\begin{abstract}
\noindent
    We study how the fusion 2-category symmetry of a fermionic (2+1)d QFT can be affected when one allows for stacking with TQFTs to be an equivalence relation for QFTs. Focusing on a simple kind of fermionic fusion 2-category described purely by group theoretical data, our results reveal that by allowing for stacking with $\mathrm{Spin}(n)_1$ as an equivalence relation enables a finite set of inequivalent modifications to the original fusion 2-categorical-symmetry. To put our results in a broader context, we relate the order of the symmetry modifications to the image of a map between groups of minimal nondegenerate extensions, and to the tangential structure set by the initial categorical symmetry on the background manifold for the QFT.
\end{abstract}

  \tableofcontents

\section{Introduction}\label{section:intro}
%%%%%%%%%%%%%%%%%%%%%%%%%%%%%%%%%%%%%%%%%%%%%%%%%%

%The subject of higher categories has been important in physics due to the role they play in the context of topological quantum field theories (TQFTs). From the functorial field theory point of view, (extended) TQFT are encoded by fully dualizable objects in a target $n$-category. 
%Furthermore, the work on understanding phases called topological orders has characterized such theories by its extended operators and their interactions, which forms a higher fusion category.

The role of higher categories has become increasingly important in physics due to its foundational role in the study of topological quantum field theories (TQFTs). Within the framework of functorial field theory, (extended) TQFTs are formally described by fully dualizable objects in a target symmetric monoidal $n$-category. In parallel, efforts to understand gapped phases of matter, called topological order,  have revealed that such phases are characterized by the structure of their extended operators. These extended operators, along with their fusion and braiding properties, organize into a higher fusion category.

Another place where higher categories arise is within the framework of our modern understanding of symmetries for quantum field theories. It is believed that the collection of symmetry operators and their composition rules assemble into some higher category. This paper aims to expand on this conventional wisdom by considering what can happen if we allow ourselves the freedom of stacking a quantum field theory by topological quantum field theories. We will give an example where the TQFT can interact with the categorical symmetry in an interesting manner, changing the structure of the symmetry, but leaving the QFT invariant.\footnote{As we will elaborate in a later section, the QFT is technically subtly modified, but the modifications cannot be detected by local observables.}
As we move into higher dimensions and gain access to nontrivial TQFTs, it becomes natural to ask whether we should broaden our notion of ``equivalence for a theory'' to go beyond the traditional scope of automorphisms and stacking by invertible TQFTs. 

The spirit of this question is frequently considered in the context of category theory.  In category theory, it is common to work with different notions of isomorphism between objects—that is, what we mean by saying two objects are equivalent depends on a chosen \textit{valent structure}, or simply ``valence''. While categories are often treated in a univalent setting, where there is a single notion of isomorphism that gives rise to equalities, it can be useful in certain contexts to retain and distinguish valent structures. Indeed, QFTs should also depend on a choice of valent structure, among other data, and different choices can lead to different consequences for the physics. 

Our main physical result demonstrates that by permitting stacking with TQFTs as a legitimate equivalence relation for a QFT, we arrive at a set of categories that are the symmetries of a particular QFT. We consider a (2+1)d setup involving a fermionic QFT i.e. one which couples to a spin structure, whose symmetry is described by a specific type of fermionic fusion 2-category. Examples where bosonic fusion 2-categories have been important for understanding the symmetry properties of physical theories have been extensively explored, see \cite{BBFP:I,BBFP,BBSNT,BSNW,DelT} for examples, and we hope that now is a good time to consider further applications of the fermionic version.
For any fusion 2-category $\mathfrak{C}$, there is a braided fusion 1-category $\Omega\mathfrak{C}$ of endomorphisms of the monoidal unit. 
Its so called M\"uger center $\mathcal{Z}_{(2)}(\Omega \mathfrak{C})$ is a symmetric fusion 1-category, and we say that $\fC$ is a \textit{fermionic fusion 2-category} if  $\mathcal{Z}_2(\Omega \mathfrak{C})$ is super-Tannakian, and bosonic if it is Tannakian. We will be working with
 fermionic fusion 2-categories $\fC$ such that $\Omega \fC$ is category $\sVect$ of super vector spaces. Such fusion 2-categories are known as \textit{fermionic strongly fusion 2-categories}\cite{Johnson-Freyd:2020ivj}. 
Physically this means that the only lines are the vacuum $1$ and the invisible (local) fermion $\psi$, however there can be nontrivial surface operators, obeying fusion rules governed by a finite group and realizing a 0-form symmetry of the theory.

 General fermionic fusion 2-categories behave subtly different from the bosonic ones. This has been noticed already in \cite[Remark 4.1.7]{decoppet2022drinfeld} when studying fermionic fusion 2-categories up to Morita equivalence. Furthermore, in the classification of fermionic fusion 2-categories, part of the parametrization data for such categories included a torsor in supercohomology  \cite[Theorem B]{Decoppet:2024htz}. This fact stems from there being no canonical choice of minimal non-degenerate extension (MNE) for $\sVect$. Motivated by this fact the main results of this work answers the following question:

\begin{question}
    What physical information is associated to the torsor in the classification of fermionic fusion 2-categories?
\end{question}

By taking the ``symmetries'' point of view for categories, which inextricably links them with quantum field theories, we can provide a physical explanation for why it is inevitable for the torsor to appear in the classification data. Conversely, the appearance of the torsor as a mathematical facet also encapsulates some subtle but essential information about the fermionic quantum field theory concerning the central charge associated to the pair $\{1,\psi\}$. 

It turns out that fermionic strongly fusion 2-categories provide a powerful and tractable framework for addressing this question.  Notably, these 2-categories offer two key advantages: they admit a relatively simple classification in terms of group-theoretic and (super)cohomological data, and they are rich enough to capture the subtle and nontrivial interplay between symmetries and stacking with TQFTs in the setting of fermionic theories. Our main focus will be to expand on the subtle interplay between the fermionic strongly fusion 2-category (the symmetry structure), the 2-category of modules of a non-degenerate braided fusion 1-category (the TQFT with condensation defects), and the spin structure (a property of the background manifold). 
Other applications of stacking by (2+1)d TQFTs with nontrivial Witt class have been considered in \cite{Bhardwaj:2024qiv,Bhardwaj:2024xcx,Bhardwaj:2025piv,Bhardwaj:2025jtf}. Looking ahead to higher dimensions, it will likely be the case that the fermionic setting will also present more subtleties due to all the data in the lower dimensional TQFTs that  needs to be included, as well as all of the interactions with the symmetries. 

%Furthermore, as the category level rises the super-Witt group of \cite{davydov2013structure} is bound to appear in the generalized cohomology theory that classifies higher categories \matt{refs}.

%%%%%%%%%%%%%%%%%%%%%%%%%%%%%%%%%%%%%%%%%%%%%%%%%%
\subsubsection*{Summary of Results}
%%%%%%%%%%%%%%%%%%%%%%%%%%%%%%%%%%%%%%%%%%%%%%%%%%

We start off with a fermionic strongly fusion 2-category $\fC$, parametrized by the data of a finite group $G_b$, a nontrivial $\Z/2$ extension given by a class $\kappa \in \H^2(BG_b;\Z/2)$, and a class in twisted supercohomology $\varpi \in \SH^{4+\kappa}(BG_b)$. As we will elaborate in Definition \ref{def:supercohclass}, a cochain representative of $\varpi$ can be written as a triple $(\alpha,\beta,\gamma)$ where $\alpha$ and $\beta$ are $\Z/2$ valued cochains and $\gamma$ is a $\mathbb{C}/\Z$ valued cohain. Each cocycle satisfies an equation with respect to the differential.
This data constructs fermionic strongly fusion 2-categories as a $G_b$-graded extension of $\tsVect$.

We choose a representation of the $\sVect$ inside $\fC$ by any of the invertible fermionic theories described by the category $\SO(n)_1$ for $n$ is an integer mod 16. As categories, $\SO(n)_1$ for each $n$ are all monoidally equivalent and comprise of the same simple objects as $\sVect$ namely 1 and $\psi$, but differ in their central charge which is given by $\frac{n}{2} \mod 8$. \footnote{In fact mathematically, often $\sVect$ is used to represent the invertible fermionic theory, but this loses out on the information of the central charge. A reason for this is because in the classification of TQFTs, the equivalence is taken up to changing the central charge i.e. up to stacking with invertible field theories.} While for a bosonic theory the choice of stacking by $\{1,\psi\}$ turns it into a fermionic one, stacking by $\{1,\psi\}$ in a fermionic theory has no affect other than to shift by a gravitational counter term.
This information is easy to miss since it can only be detected by placing the theory on curved background, and would depend on what one takes as the condition for two QFTs to be equivalent. We iterate the following symmetry modifying procedure, which we will refer to as ``stack and condense'' later on in this paper: 

\begin{construction}\label{construction:F2Cmodify}
\begin{enumerate}
    \item On a theory $\mathcal{T}$ with a fermionic strongly fusion 2-category symmetry $\fC$, we stack by the 2-category of modules of a minimal non-degenerate extension of $\sVect$ given by $\Spin(n)_1$ for some $n$. This is a category with three objects if $n$ is odd and four objects when $n$ is even. In both cases there is a simple object  with quantum dimension 1 which we denote $f$, that has spin $\frac{1}{2}$. By taking the composite $\psi \boxtimes f$ this creates a boson, which we condense.
    \item Upon condensing the boson we arrive at a different realization of the symmetry $\fC$. We then compute the shift of the cochain representation of $\varpi$ due to the condensation. The shift must be compatible with how $\varpi$ behaves with respect to acting by a differential. 
\end{enumerate}
\end{construction}

The two steps can be summarized in Figure \ref{fig:twostep}:
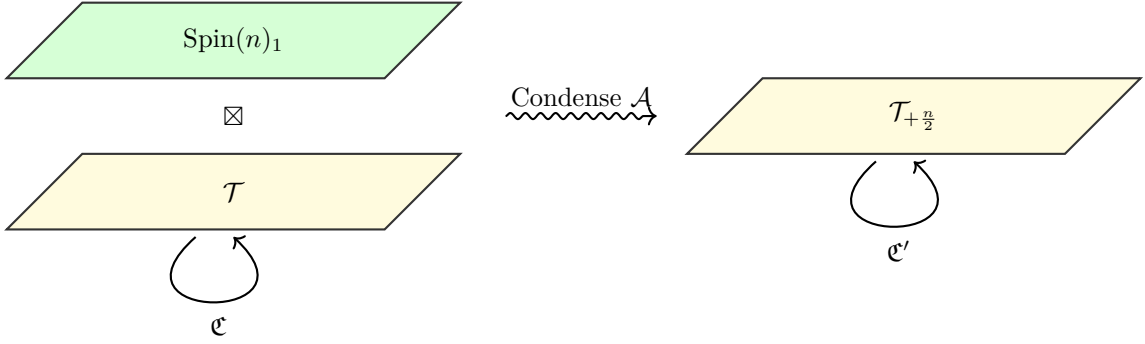
\begin{figure}
    \centering
\begin{tikzpicture}[thick]
\draw[black, fill=yellow!20,opacity=0.8] (3,1) -- (8,1) -- (9,2) -- (4,2) -- cycle;
\draw[black, fill=green!20,opacity=0.8] (3,3) -- (8,3) -- (9,4) -- (4,4) -- cycle;
\node at (6,1.5) {$\mathcal{T}$};
\node at (6,2.5) {$\boxtimes$};
\node at (6,3.5) {$\Spin(n)_1$};
    \draw[<-,decorate]  (6,.9) .. controls (7.25,-.25) and (4.2,-.25) .. (5.5,.9);
      \node at (5.8,-0.28) {$\mathfrak C$};
        \draw[<-,decorate]  (6+9,.9+1) .. controls (7.25+9,-.25+1) and (4.2+9,-.25+1) .. (5.5+9,.9+1);
      \node at (5.8+9,-0.28+1) {$\mathfrak{C}'$};
 \draw[->,decorate,decoration={snake,amplitude=.4mm,segment length=2mm,post length=1mm}] (9.6,2.5) -- (11.6,2.5) node[midway, above] {Condense $\mathcal{A}$};
 \draw[black, fill=yellow!20,opacity=0.8] (3+9,1+1) -- (8+9,1+1) -- (9+9,2+1) -- (4+9,2+1) -- cycle;
 \node at (6+9,1.5+1) {$\mathcal{T}_{+\frac{n}{2}}$};
\end{tikzpicture}
    \caption{Left: The figure shows a stacking of a $\Spin(n)_1$ TQFT onto a fermionic (2+1)d QFT  given by $\mathcal T$, which has an action by the fermionic strongly fusion 2-category $\fC$.
    Right: After condensing the algebra $\mathcal{A} = 1\boxtimes1 \oplus (\psi \boxtimes f)$ the TQFT is trivialized and the theory $\cT$ is the same up to a shift of the central charge coming from $\Spin(n)_1$ i.e. a gravitational counter term. We thus denote the theory after condensing $\cA$ to be given by $\cT_{+\frac{n}{2}}$ in order to reflect this fact. $\cT_{+\frac{n}{2}}$ then has an action by the category $\mathfrak{C}'$ which has the same operators as the fusion 2-category $\fC$, but with a potentially different monoidal structure.}
    \label{fig:twostep}
\end{figure}

\begin{maintheorem}(Theorem \ref{prop:shift})\label{prop:mainA}
    The periodicity of $\varpi$ depends on the following properties of $G_b$ and $\kappa$.
    \begin{itemize}
    \item if $\kappa= 0$: $\varpi$ experiences no shift under the stacking and condensing procedure;
        \item  if $\kappa\neq 0$ and $\Sq^1\kappa=0$: $\varpi$ experiences a 2-periodic shift;\footnote{In particular a single iteration sends
        \begin{equation*}
            (\alpha, \beta, \gamma) \mapsto (\alpha+\kappa, \,\beta+\kappa \cup_1 \alpha, \, \gamma+G_\kappa(\alpha))
        \end{equation*}
        where $G_\kappa(\alpha)$ is some $\mathbb{C}/\Z$ valued function such that $G_\kappa(\alpha)+G_\kappa(\alpha+\kappa) =0$.}
        \item  if  $\kappa \neq 0$ and $\Sq^1 \kappa \neq 0$: $\varpi$ experiences a  4-periodic shift.
    \end{itemize}
\end{maintheorem}
   
 At the end of each iteration of stacking and condensing the theory $\cT$ remains the same since the TQFT that was stacked on has been condensed away, but the TQFT does contribute a gravitational counter term due to the fact that the $\Spin(n)_1$ theory that was stacked has a nontrivial central charge.
QFTs that differ by ``central charge'' are isomorphic, and hence the QFTs from before and after stacking by $\Spin(n)_1$ and condensing are equal in the sense that the TQFT has not changed any of the operator content of the original theory.
However, the symmetries are different due to the cocycle representative of the cohomology class being genuinely different. As we will explain in \S\ref{section:fermionicsym}, $\alpha$ controls the extension of $G_b$ with $\Z/2$, so by setting $\kappa=\alpha$ then we can trivialize the extension and hence change the fusion rules of the symmetry defects.
Thus we conclude that the same fermionic theory can have a set of inequivalent descriptions of its categorical symmetries, and the TQFTs acts on the operator content so in a way that they transform under a different symmetry.

The second part of our results puts the shift of  $\varpi$ into a larger context, with the aim to offer a conceptual explanation for the periodicity of the shifts. We consider the (3+1)d TQFT associated to fermionic strongly fusion 2-categories. The one specifically associated to $\tsVect$, is a TQFT of a dynamical spin structure, and constructed out of the Drinfeld center $\cZ(\tsVect)$. This theory was essential to the proof of the existence of minimal non-degenerate extensions \cite{johnson2024minimal}. In particular it is used to show that a braided fusion 1-category $\cB$ with Müger center $\cZ_{(2)}(\cB)\cong\sVect$ has a MNE. This result generalizes to the case when the Müger center is super-Tannakian. In this case the TQFT is associated to a theory of a dynamical spin structure, along with a $G_f$ gauge symmetry; this is sometimes called a spin$\text{-}G_f$ gauge theory; see Example \ref{def:twistedspinex2} for a definition. \footnote{In this language, the theory given by $\cZ(\tsVect)$ could be called a spin$\text{-}\Z/2$ gauge theory.} This notation means that the common $\Z/2$ of the group $\Spin_n$ and $G_f$ are identified. Using spin$\text{-}G_f$ to describe the gauge theory is particularly evocative because it makes contact with twisted spin structures.
The Lagrangian descriptions for these gauge theories are given by classes in $\SH^{4+\kappa}(BG_b)$ but they are not canonical, as pointed out in in \cite{Johnson-Freyd:2022}. 

%Proposition \ref{prop:mainA} allows us to argue that actually the TQFT associated to $[(\alpha,\beta,\gamma)] \in \SH^{4+\kappa}(BG_b)$ is actually only defined with $\alpha \in \rH^2(BG_b;\Z/2)$ modulo $\kappa$ hence  spelling out in detail a point presented in \cite{Johnson-Freyd:2022}.

The (3+1)d TQFT of a dynamical spin structure has automorphisms given by the group $\Z/16 \cong \mathrm{Mext}(\sVect)$ of MNEs of $\sVect$ as shown in \cite[Equation 19]{Johnson-Freyd:2020twl}. These automorphisms  are implemented by stacking the boundary of the TQFT with $\Spin(n)_1$. Hence, for a super-Tannakian category $\cE = \Rep(G_f)$ the group $\Mext(\cE)$ should be associated to the automorphisms of the spin$\text{-}G_f$ theory. Some computations of $\Mext(\cE)$  have been done in \cite[Section 5]{nikshych2022computing} for when $G_f = \Z/2^{n-1} \textbf{.}_{\kappa} \Z/2$ i.e.\ an extension of groups, and  $G_f=\Z/2 \times \Z/2$. There is a map $F: \Mext(\cE) \rightarrow \Mext(\sVect)$, which is a homomorphism called the central charge map, and the  image of the map $F$ is also computed in \cite[Corollary 5.5, Lemma 5.9]{nikshych2022computing}; we elaborate more on this map, including its construction in \S\ref{section:TQFTandSym}.
Thanks to the relationship between minimal non-degenerate extensions and (3+1)d TQFT arising from fusion 2-categories, we can make contact between the map $F$ and the computation in Theorem \ref{prop:mainA}.
It turns out that a way of interpreting our computation for the periodicity of $\varpi$ is that it reveals how many of the 16 $\Spin(n)_1$ theories in the image of $F$ leave the cocycle invariant. Another way of viewing the cases when $\varpi$ does not shift is to consider which of the $\Spin(n)_1$ theories can be lifted from the boundary of the TQFT $\cZ(\tsVect)$ to the boundary of the spin$\text{-}G_f$ TQFT, i.e. coupling it to a background twisted spin structure, see Figure \ref{fig:composeSandwich}. A form of this question was studied in \cite{Barkeshli:2021ypb} where the authors classified invertible fermionic theories. Using part of the data in their classification  we can also reproduces the known results for the image of the map $F$, for different $\cE$.

Our next proposition puts the results of Theorem \ref{prop:mainA},\cite{nikshych2022computing}, and \cite{Barkeshli:2021ypb}, into a unified setting:

%The conditions on $\kappa$ that are needed for $\varpi$ to be 2 or 4 periodic,  exactly matches with the image of the map $\Mext(\cE)\rightarrow\Mext(\sVect)$. What this map does is it takes a theory that was compatible on a spin$\text{-}G_f$ background and forgets the extra $G_f$ symmetry.
% our results reproduce those in \cite{nikshych2022computing}.

\begin{maintheorem}(Theorem \ref{prop:TFAE})\label{prop:mainB}
 Let $\fC$ be a fermionic strongly fusion 2-category parametrized by the following data. A finite group $G_b$, a class $\kappa \in \rH^2(BG_b;\Z/2)$ such that $G_f = G_b\,  \textbf{.}_{\kappa} \Z/2$, a class $\varpi \in \SH^{4+\kappa}(BG_b)$. Let $\cE = \Rep(G_f)$, then the following are equivalent:
 \begin{enumerate}
     \item  \label{point1} The $\Spin(n)_1$ theories such that application of Construction \ref{construction:F2Cmodify} on $\cT$ stacked with $\Spin(n)_1$ leaves $\varpi$ invariant.
     \item \label{point2} The image of $F:\Mext(\cE) \to \Mext(\sVect)$.
     \item \label{point3} The $\Spin(n)_1$ that can be coupled to background spin$\text{-}G_f$ structure set by the symmetry $\fC$.
 \end{enumerate}
\end{maintheorem}
\begin{rem}
   By Theorem \ref{prop:mainA}, the shifts of $\varpi$ in (\ref{point1}) depend only on conditions for $\kappa$, and applies for all groups $G_f$ that can be constructed from $\kappa$ and $G_b$. Therefore in order to make contact between (\ref{point1}) and  (\ref{point3}) we must consider all possible spin$\text{-}G_f$ backgrounds for $\Spin(n)_1$. Suppose there exists some $G_b$ for which the associated $G_f$ is nontrivial, but $\Spin(n)_1$  for some $n$ does not satisfy the consistency conditions in \cite[Table 1]{Barkeshli:2021ypb}. Then this theory cannot be consistent on the symmetry background, and therfore cannot be used for the stacking procedure in  (\ref{point1}) onto $
   \fC$.
\end{rem}

%%%%%%%%%%%%%%%%%%%%%%%%%%%%%%%%%%%%%%%%%%%%%%%%%%
\subsubsection*{Outline}
%%%%%%%%%%%%%%%%%%%%%%%%%%%%%%%%%%%%%%%%%%%%%%%%%%
The structure of this paper is as follows. In \S\ref{section:fermionicsym} we start by giving some background for fermionic fusion 2-categories and the necessary cohomological information. Then in \S\ref{subsection:shifting}  we prove Theorem \ref{prop:mainA}, hence showing how the TQFT interacts with the fusion 2-category symmetry. In \S\ref{section:TQFTandSym} we explain the mutual interactions between categorical symmetries and TQFT. We also relate the results of Theorem \ref{prop:mainA} with minimal modular extensions, culminating in a proof of Theorem \ref{prop:mainB}. 

%%%%%%%%%%%%%%%%%%%%%%%%%%%%%%%%%%%%%%%%%%%%%%%%%%
\subsection*{Acknowledgments}
%%%%%%%%%%%%%%%%%%%%%%%%%%%%%%%%%%%%%%%%%%%%%%%%%%

It is a pleasure to thank Theo Johnson-Freyd, Ryohei Kobayashi, Ho Tat Lam, Dmitri Nikshych, Abhinav Prem, Xiao-Gang Wen, and Yunqin Zheng for helpful conversations. 
This research was supported in part by grant NSF PHY-2309135 to the Kavli Institute for Theoretical Physics (KITP). Part of this work
was completed at the  KITP and benefited from interactions with the participants of the  program “Generalized
Symmetries in Quantum Field Theory: High Energy Physics, Condensed Matter, and Quantum
Gravity”. The author M.Y. is supported by the EPSRC Open Fellowship EP/X01276X/1.
The author D.T. is supported by the Simons Collaboration for Global Categorical Symmetries and the Nova Scotia Graduate Scholarship program. No authors have competing interests to declare that are relevant to the content of this article. There is no data available for this article to declare.

%%%%%%%%%%%%%%%%%%%%%%%%%%%%%%%%%%%%%%%%%%%%%%%%%%
\section{Describing the Fermionic Categorical Symmetry}\label{section:fermionicsym}
%%%%%%%%%%%%%%%%%%%%%%%%%%%%%%%%%%%%%%%%%%%%%%%%%%

\subsection{Preliminaries}\label{subsection:prelim}
We begin by setting up the relevant definitions and classification data related to fermionic fusion 2-categories. The main goal of this subsection is to explain how the twisting of supercohomology in the sense that we are going to define, can be viewed in the same way as a twisted spin structure.  This perspective will prove useful in the analysis carried out in \S\ref{section:TQFTandSym}.

For fermionic categories with grouplike symmetries, we define a fermionic group as follows.
\begin{defn}
\label{def:fermionic_symmetry}
A \textit{fermionic symmetry} is the data of a group $G_b$, and a central extension
\begin{equation}
    1\rightarrow \Z/2\rightarrow G_f\rightarrow G_b\rightarrow 1\,,
\end{equation}
given by a class $\kappa \in H^2(BG_b;\Z/2)$.
We refer to $G_b$ as the \textit{bosonic symmetry group} of the fermionic symmetry.\footnote{The most general version of a fermionic group discussed in \cite[\S 7]{Ben88}, and \cite{Debray:2023iwf},  includes an anti-unitary component, that involves a class in $\rH^1(BG_b;\Z/2)$.}
\end{defn}
As is conventional, we will refer to a fermionic symmetry in terms of its full fermionic symmetry group $G_f$.
The categorical generalization of the fermionic group that acts on QFTs in (2+1)d is a fermionic fusion 2-category. By a fusion 2-category we mean a semisimple monoidal 2-category with a simple monoidal unit introduced in \cite{douglas2018fusion},
and we say that a fusion 2-category $\fC$ is fermionic if $\cZ_{(2)}(\Omega \fC)$ is super-Tannakian. We will use the notation that $\Omega\fC  \cong \End_{\fC}(\mathds{1})$ and the adjoint of $\Omega$ is given by a ``suspension map'' $\Sigma$, for which  $\Sigma \fC$ is the same as taking the category of modules of $\fC$. 

\begin{defn}{\cite{Johnson-Freyd:2020ivj}}
    A fermionic strongly fusion 2-category  is a fusion 2-category $\fC$ such that $\Omega \fC \cong \sVect$. 
\end{defn}

\begin{theorem}{\cite[Theorem B]{Johnson-Freyd:2020ivj}}
    If $\mathfrak C$ is a fermionic strongly fusion 2-category, the equivalence classes of indecomposable objects of $\mathfrak C$ form a finite group, which is a central double cover of the group $\pi_0 \mathfrak C = G_b$ of components of $\mathfrak C$. 
\end{theorem}

\begin{example}
    The simplest example of a nontrivial fermionic strongly fusion 2-category is given by the trivial extension $\tsVect \boxtimes \mathbf{2Vect}^\pi_{G_b}$ where $\pi \in \H^4(BG_b;\mathbb C^\times)$, for which the group $G_f = \Z/2\times  G_b$, i.e. the $\H^2(BG_b;\Z/2)$ class giving rise to the extension is trivial. Since the group $\H^2(BG_b;\Z/2)$  is trivial when $G_b$ is odd, in this case all the fermionic strongly fusion 2-categories take the above form.
\end{example}

Since a fermionic strongly fusion 2-category feels very much like a group $G_f$, it is useful to think of this type of categorical symmetry as arising from a $G_f$ bundle which gives rise to a twisted spin structure for the underlying manifold.
\begin{defn}
\label{def:twisted}
     Let $V\to X$ be a vector bundle. An \emph{$(X, V)$-twisted spin structure} on a vector bundle $E\to M$ is data of a map $f: M\to X$ and a spin structure on $E\oplus f^*(V)$. %\arun{Allow virtual bundles, so \pinp is $1-\sigma$ rather than $3\sigma-3$?}
\end{defn}
The last condition  means that a twisted spin structure requires a trivialization of $w_1(E\oplus f^*(V))$ and $w_2(E\oplus f^*(V))$.
%\footnote{While we will focus on only the case where the twist is given by vector bundles, in general it was found in \cite{Debray:2022wcd} and proven in \cite{Debray:2023tdd} that it is possible to consider twists of spin structure by cohomology classes that do not arise as characteristic classes of some vector bundle.}
In order to be compatible with the fermionic group in Definition \ref{def:fermionic_symmetry} of a fermionic fusion 2-categories, which were defined using $\kappa$, we will not consider the full notion of a twisted spin structure that includes twists by classes in $\H^1(X;\Z/2)$ that modify $w_1(E)$ and hence the orientation of the underlying manifold.
For $b\in \H^2(X;\Z/2)$ we will take $V$ such that $w_2(V) = b$ and thus $w_2(E) = w_2(V)$ by the Whitney sum formula.

\begin{example}\label{def:twistedspin}
    In terms of a twisted spin structure, spin$^c$ is given by a $(B\mathbb{C}^\times,L)$-twisted spin structure, where $L$ is the complex line bundle over $B\mathbb{C}^\times$.  By Definition \ref{def:twisted}, for $E= TM$ this gives an identification $w_2(TM) = c_1 \mod 2$.
\end{example}

\begin{example}\label{def:twistedspinex2}
    Consider the space $B\Z/2^{n}$ where $n> 2$. The cohomology is $\rH^*(B\Z/2^{n};\Z/2) = \Z/2[x,y]/\{x^2=0\}$ with $|x|=1$ and $|y|=2$. Let $W \to B\Z/2^{n}$ be the rotation representation with $w_2(W) = y$. A spin$\text{-}\Z/2^n$ structure is given by $$ \Spin \times_{\{\pm 1\}} \Z/2^n: = \frac{\Spin \times \Z/2^n}{\{\pm 1\}}$$
    where the common $\Z/2$ quotient is fermion parity, and $\Spin:= \lim_{n\to \infty} \Spin_n$ is the infinite spin group. This structure fits into the following pullback:
    \begin{equation}\label{eq:GMpullback}
    \begin{tikzcd}[column sep=2cm, row sep=2cm]
\Spin \times_{\{\pm 1\}} \Z/2^n \arrow[r] \arrow[d] \arrow[dr, phantom, "\lrcorner", very near start] & \Z/2^{n-1} \arrow[d, "y"] \\
\SO \arrow[r, "w_2(TM)",swap] & B\Z/2\,.
\end{tikzcd}
\end{equation}
which implies an identification of $w_2(TM)$ and $y$. Hence, a spin$\text{-}\Z/2^n$ structure is equivalent to a $(B\Z/2^{n-1},W)$-twisted spin structure. A spin$\text{-}G_f$ structure is defined in an analogous way as a twisted spin structure, whose twisting is given by a rank 2 bundle over the space $BG_b$.
We will call upon this example in \S\ref{section:TQFTandSym}.
\end{example}

Twisted tangential structures play well with generalized cohomology, as instead of computing twisted generalized cohomology of a space, one can equivalently compute the (untwisted) generalized cohomology of a Thom spectrum. Intuitively, one should think that the Thom spectrum as an all-in-one object that captures the space in question along with the twist. For $V\to X$  a real vector bundle with Euclidean metric, the \emph{Thom space} $\mathrm{Th}(X,V)$ is the quotient $D(V)/S(V)$, where $D(V)$ is the disk bundle of vectors in $V$ with norm less than or equal to 1, and $S(V)$ is the sphere bundle of vectors of norm equal to 1.  The \emph{Thom spectrum} $X^V$ of a vector bundle $V\rightarrow X$ is the suspension spectrum of the Thom space $\Sigma^\infty \mathrm{Th}(X,V)$. 

\begin{example}\label{ex:bordism}
     Using the groups in  \ref{def:twistedspinex2}, if one wants to compute bordism for manifolds with the twisted spin structure spin$\text{-}\Z/2^n$, then it is equivalent to computing the homotopy groups of 
\begin{equation}
    MTSpin \wedge (B\Z/2^{n-1})^{W-2}\,,
\end{equation}
which give $\Omega^\Spin_*((B\Z/2^{n-1})^{W-2})$.
\end{example}

 While our discussion has centered on twisted spin structures arising from vector bundle, this is not strictly necessary. In general, for $X$ a space, with $a\in H^1(X;\Z/2)$, and $b \in H^2(X;\Z/2)$. The {$(X, a, b)$-twisted spin structure} on a vector bundle $E\to M$ is data of a map $f\colon M\to X$ and trivializations of $w_1(E) + f^*(s)$ and $w_2(E) + w_1(E)^2 + f^*(\omega)$, but we do not assume there is a vector bundle $V \rightarrow X$ such that $s = w_1(V)$ and $b = w_2(V)$. This type of twisted spin structure which just considers cohomology classes was coined the ``fake vector bundle'' twist in  \cite{Debray:2023tdd} see also \cite{Debray:2022wcd} for an application. It is still possible even in the fake vector bundle case to define Thom spectra,  following \cite{ABGHR14a,ABGHR14b}. 
%This was done via a shearing construction in \cite{DDHM23,Debray:2023ior,Debray:2025iqs} in the context of bordism, where the generalized cohomology theory is represented by a particular Madsen-Tillmann spectrum of manifolds with some $\eta$-tangential structure. 
%As it will be useful in \S\ref{section:TQFTandSym} we briefly review how this procedure works here.
%Consider $V\to X$ a real vector bundle, with a choice of Euclidean metric. We denote by $D(V)$ be the disk bundle of vectors in $V$ with norm less than or equal to 1, and by $S(V)$ the sphere bundle of vectors of norm equal to 1.

%\begin{defn}
%     For $V\to X$  a real vector bundle with Euclidean metric, the \emph{Thom space} $\mathrm{Th}(X,V)$ is the quotient $D(V)/S(V)$.
%\end{defn}

%\begin{defn}
 %   The \emph{Thom spectrum} $X^V$ of a vector bundle $V\rightarrow X$ is the suspension spectrum of the Thom space $\Sigma^\infty \mathrm{Th}(X,V)$.
%\end{defn}

%\begin{lemma}[Shearing]\label{lem:shearing}
%Let $\xi: B\rightarrow B\O$ be a tangential structure.
   % Let $S \rightarrow B\O$ be the tautological rank-zero virtual vector bundle and $\eta: B \times X \to B\O$ a tangential structure classified by the virtual vector bundle $\xi^*(S)\boxplus (V-r_V)$ where $r_V$ denotes the rank of $V$. Then an $\eta$-structure is equivalent to a $(X,V)$-twisted $\xi$-structure.
%\end{lemma}
This more generalized notion of twisted spin structures by cohomology classes will now assist in understanding twisted 
supercohomology, which is a generalized cohomology theory that has seen many applications in physics to fermionic topological phases \footnote{It was also shown in \cite{Debray:2023tdd} that restricted supercohomology can be used to twist string structures, in the context of heterotic string theory.}. Computations of twisted supercohomology are similar in spirit to Example  \ref{ex:bordism}, and can be done by using a Thom spectrum.

We now present one way of introducing supercohomology. We start with a ``reduced'' version of supercohomology $\mathrm{rSH}$ introduced in \cite{Freed:2006mx,Gu:2012ib} which has
 \begin{equation}
     \pi_0(\mathrm{rSH})=\bC^\times, \quad  \pi_{-1}(\mathrm{rSH})=\Z/2\,,
 \end{equation}
  and the $k$-invariant is given by
  \begin{equation}
      (-1)^{\Sq^2}: \rH^{*}(-;\Z/2) \to \rH^{*+2}(-;\bC^\times).
  \end{equation}
   Here, $\Sq^{2^n}:\rH^*(-;\Z/2) \to \rH^{*+2^n}(-;\Z/2) $ is a  Steenrod square, a stable cohomology operation, and $(-1)^x:\rH^*(-;\Z/2) \to \rH^*(-;\Z/2)$ stands for the operation induced by the inclusion $\Z/2 \hookrightarrow \bC^\times$, which is nontrivial on those elements in $\rH^*(-;\Z/2)$ that have  nontrivial image under $\Sq^1$ in $\rH^{*+1}(-;\Z/2)$.

There is also ``extended'' supercohomology $\SH$ used in \cite{Kapustin:2017jrc,Gaiotto:2017zba,Johnson-Freyd:2021tbq,Decoppet:2025eic} which has the homotopy groups of $\mathrm{rSH}$ but with and extra homotopy group $\pi_{-2}(\SH)=\Z/2$, reflecting the fact that the Majorana chain is an invertible (1+1)d phase. The $k$-invariant between $\pi_{-2}$ and $\pi_{-1}$ is given by $\Sq^2$.
When we refer to supercohomology we will always be refer to the extended version of supercohomology. 
\begin{rem}
    The way of realizing supercohomology that was just explained can be seen as the dual of superhomology which we denote $\tau_{\leq 2} \ko$, which can be viewed as a truncation of $\ko$ by the following short exact sequences:
\begin{equation}
        \tau_{\geq 3} \ko \longrightarrow  \ko \longrightarrow \tau_{\leq 2} \ko\,.
\end{equation}
where $\tau_{\geq k}$ denotes the homotopy groups in degree $k$ and above. Using this short exact sequence it is possible to develop a modified Adams spectral sequence to conduct computations for twisted $\tau_{\leq 2} \ko$-homology.
Examples of computations with applications to anomalies, appear in \cite{DYY1}.
\end{rem}

We now present supercohomology in a slightly different context where it is used to parametrize fermionic fusion 2-categories.  We start by noting that the homotopy groups of $\tsVect^\times$ are given by 
\begin{equation}\label{eq:homotopyBrPicsvect}
 \begin{tabular}{|c|c|c|c|}
\hline
$\pi_0$ & $\pi_1$ & $\pi_2$  \\
\hline \\[-1.3em]
 $\Z/2 $ & $\Z/2$ & $\bC^\times$ \\
\hline
\end{tabular}
\end{equation}
which looks like the homotopy groups of $\SH$ shifted by 2.
In fact, for a space $X$, the homotopy classes of maps  $X \to B^{n-2}\tsVect^\times$ are parametrized by 
$\SH^n(X)$. 

Importantly, the groupoid $\tsVect^\times$ furthermore has a non-trivial
automorphism given by fermion parity, i.e. $B\Z/2 \simeq \mathcal{A}ut^{br}(\sVect)$, which acts by $-1$ on the odd vector spaces. The action of $B\Z/2$ on $\tsVect^\times$ is given by the following fiber sequence
\begin{equation}\label{eq:fibration}
    \begin{tikzcd}
        B^{n-2}\tsVect^\times \arrow[r] & B^{n-2}\tsVect^\times/\!/ (B\Z/2) \arrow[d]\\
        & B^2 \Z/2\,.
    \end{tikzcd}
\end{equation}

\begin{defn}
    Let $X$ be a space equipped with a map $\kappa:X\to  B^2\Z/2$. The $\kappa$-twisted $n$-th cohomology of $X$ is the group $\SH^{n+\kappa}(X)$ of homotopy classes of
    $B\Z/2$-equivariant maps from $X$ to $B^{n-2}\tsVect^\times$.
\end{defn}

\begin{defn}\label{def:supercohclass}
   A $n$-cocycle in twisted supercohomology is the data of a triple $(\alpha, \beta, \gamma)$ where:
    \begin{itemize}
        \item $\alpha \in \rH^{n-2}(BG_b;\Z/2)$ is the \textit{Majorana layer}, 
        \item $\beta \in \mathrm{C}^{n-1}(BG_b;\Z/2)$ is the $\textit{Gu-Wen}$ layer, such that $d\beta =(\Sq^2+\kappa) \alpha$,
        \item and $\gamma \in \mathrm{C}^n(BG_b;\mathbb{C}/\Z)$ is the \textit{Dijkgraaf-Witten} layer, such that $d\gamma = i((\Sq^2+\kappa)\beta)+ f_{\kappa}(\alpha)$
        where $i:\Z/2 \to \mathbb{C}/\Z$ and $f_\kappa(\alpha)$ is a cochain that corrects the failure of $(-1)^{(\Sq^2+\kappa)\beta}$ to be closed. 
    \end{itemize}  
\end{defn}
\noindent Since we adopt the additive notion in the third bullet, we will denote the coefficients by $\mathbb{C}/\Z$, rather than $\mathbb{C}^\times$ for the cochain. For an explicit description of $f_\kappa(\alpha)$, see \cite[Appendix C]{Barkeshli:2022edm}, where it is denoted by $O_5$ therein.
Looking ahead, the Majorana layer will play a significant role in this work, and is what leads to the subtleties of supercohomology as well as the shifts in $\varpi$.

Since the fully fledge theory of fusion 2-categories does not concern us, we will now explain  how twisted supercohomology arises in just the case of fermionic strongly fusion 2-categories. By starting off with $\tsVect$, one way to construct fermionic strongly fusion 2-categories $\fC$ with $\pi_0 \fC = G_b$ is with extension theory \cite{etingof2010fusion}. As was shown in \cite[Theorem 3.11]{decoppet2024extension} faithfully graded $G_b$ extensions of $\tsVect$ are parametrized by maps 
\begin{equation}
    BG_b \rightarrow B \sBrPic(\tsVect),
\end{equation}
where $\sBrPic(\tsVect)$ is the Brauer-Picard groupoid of $\tsVect$ and consists of the
invertible objects and the invertible morphisms in the monoidal 3-category of finite semisimple $\tsVect$-bimodule 2-categories. Physically, this describes the invertible surface and line operators in the Drinfeld center of $\tsVect$. We first present the homotopy groups of $\sBrPic(\tsVect)$, which are given by \cite[Section 4.1]{decoppet2024extension}:
\begin{equation}\label{eq:homotopyBrPic}
 \begin{tabular}{|c|c|c|c|c|}
\hline
$\pi_0$ & $\pi_1$ & $\pi_2$ & $\pi_3$ \\
\hline \\[-1.3em]
1 & $\Z/2 \oplus \Z/2$ & $\Z/2$ & $\bC^\times$ \\ 
\hline
\end{tabular}\,.
\end{equation}
By this description, it is straightforward to see that the space $B\sBrPic(\tsVect)$ has nontrivial homotopy groups in degrees 2, 3, and 4. The  $k$-invariants are given in \cite{Gaiotto:2017zba,Johnson-Freyd:2020twl}, and shown to be the same as the ones for supercohomology. They can be understood by passing to the center of $\tsVect$ and using the equivalence
 $B\sBrPic(\tsVect) \simeq B^2 \cZ(\tsVect)^\times$, see \cite[Lemma 4.1]{decoppet2024extension}. By using \cite[Lemma 4.4]{decoppet2024extension}, which provides an isomorphism of the fibration $B^2\tsVect^\times/\!/(B\Z/2) \rightarrow B^2\Z/2$ in Equation \eqref{eq:fibration} with $B^2\cZ(\tsVect)\rightarrow B^2\Z/2$, we find that
the composite map $\kappa$
 \begin{equation}
 \begin{tikzcd}
  BG_b \arrow[r] \arrow[dotted,rd,"\kappa",swap] & B^{2}\cZ(\tsVect)^\times \arrow[d]\\
        & B^2 \Z/2\,,   
 \end{tikzcd}
 \end{equation}
endows $BG_b$ with an action by $B^2\Z/2$. Then a map  $BG_b \rightarrow B^2\cZ(\tsVect)^\times$ that lifts $\kappa$ restricts to $B^2\tsVect^\times$ and is equivalent to a class in $\SH^{4+\kappa}(BG_b)$.

We can now state the classification of fermionic stronly fusion 2-categories.
\begin{proposition}\label{prop:stronglyF2C}\cite{decoppet2024extension}
Fermionic strongly fusion 2-categories are classified by a bosonic group $G_b$, a class $\kappa \in \H^2(BG_b;\Z/2)$, and a class $\varpi \in \SH^{4+\kappa}(BG_b)$.
\end{proposition}
\begin{rem}
    Given the classification data in Proposition \ref{prop:stronglyF2C}, a fermionic strongly fusion 2-category naturally gives rise a twisted spin structure in Definition \ref{def:twistedspin}. In the case where we have a bundle $V\to X$ such that $w_2(V) = \kappa$, then we can form the Thom space $X^V$, and  the twisted supercohomology $\SH^{n+\kappa}(X)$ is equivalent to $\SH^n(X^V)$. In the case when there is no bundle, we use the fake vector bundle twist to arrive at twisted supercohomology of $X$.
\end{rem}

\begin{example}
    When $\alpha$ is trivial, then there is no obstruction to gauge the $\Z/2$ fermion parity symmetry in the fermionic QFT. Upon gauging, the dual symmetry will be a $\Z/2$ one-form summetry. The corresponding strongly fusion 2-category can be represented by $\mathbf{2Vect}^\pi_{\mathcal{G}}$ for $\mathcal{G} = B\Z/2 \textbf{.}_{\beta}\, G_b$ i.e. a fusion 2-category of $\pi$-twisted 2-group graded 2-vector space for $\pi \in \rH^4(B\mathcal{G};\mathbb{C}^\times)$.
    The Gu-Wen layer can now be thought of as the Postnikov class that gives rise to the 2-group extension.
    In the case of (3+1)d topological field theories, this corresponds to the statement that when the Majorana layer is trivial, then it is always possible to gauge fermion parity and arrive at a ``true'' gauge theory for a higher group, i.e. a Dijkgraaf-Witten theory where the action is given by $\pi$ \cite{Johnson-Freyd:2022}.
\end{example}

%%%%%%%%%%%%%%%%%%%%%%%%%%%%%%%%%%%%%%%%%%%%%%%%%%
\subsection{Shifting the cocycle}\label{subsection:shifting}
%%%%%%%%%%%%%%%%%%%%%%%%%%%%%%%%%%%%%%%%%%%%%%%%%%
 We now explain the process of stacking  by the $\Spin(n)_1$ TQFTs, condensing, and how it affects the supercohomology cocycle of the strongly fusion 2-category symmetry. This will be employing the well known Morita equivalence in fusion 2-categories between $\tsVect \boxtimes \Sigma \cC$ and $\tsVect$, where $\cC$ is a minimal modular extension of $\sVect$. We will make a choice to represent categories in this Morita equivalence with (2+1)d topological field theories.
 First off, we take $\cC$ to be given by $\Spin(n)_1$ where $n$ is an integer modulo 16. The theories $\Spin(n)_1$ has chiral central charge  $c \equiv \frac{n}{2} \mod 8$ as well as the following data:

 \begin{itemize}
     \item for $n\equiv 0 \mod 4$: 
     \begin{equation}\label{eq:Nevendata}
     \begin{tabular}{|c|c|c|c|c|c|}
\hline
$\text{Objects}$ & $1$ & $f$ & $e$ & $m$\\
\hline \\[-1.3em]
$\text{Spin}$ & $0$ & $\frac{1}{2}$ & $\frac{n}{16}$ & $\frac{n }{16}$\\
\hline
\end{tabular}
  \end{equation}
  The fusion rules are given by:
  \begin{equation}
      f \times f = e \times e= m \times m=1, \quad e\times m = f, \quad m \times f =e, \quad f \times e =m\,.
  \end{equation}
  There are three $\Z/2$ one-form symmetries generated by $f,e,m$. 
  \item for $n\equiv 2 \mod 4$: 
   \begin{equation}\label{eq:Nmod2data}
     \begin{tabular}{|c|c|c|c|c|c|}
\hline
$\text{Objects}$ & $1$ & $f$ & $a$ & $\overline{a}$\\
\hline \\[-1.3em]
$\text{Spin}$ & $0$ & $\frac{1}{2}$ & $\frac{n}{16}$ & $\frac{n }{16}$\\
\hline
\end{tabular}
  \end{equation}
  The fusion rules are given by:
  \begin{equation}
      f \times f = a \times \overline{a}=1, \quad a\times {a} =\overline{a}\times \overline{a} =  f, \quad a \times f =\overline{a}, \quad \overline{a} \times f =a\,.
  \end{equation}
  
  The one-form symmetry is $\Z/4$ with the $\Z/2$ subgroup generated by $f$.
     \item for $n\equiv 1 ,3 \mod 4$: 
      \begin{equation}\label{eq:Nodddata}
     \begin{tabular}{|c|c|c|c|c|c|}
\hline
$\text{Objects}$ & $1$ & $f$ & $\sigma$ \\
\hline \\[-1.3em]
$\text{Spin}$ & $0$ & 1 & $\frac{n}{16}$\\
\hline
\end{tabular}
  \end{equation}
  The fusion rules are given by:
  \begin{equation}\label{eq:fusionodd}
      f\times f = 1,\quad f\times \sigma = \sigma \times f = \sigma, \quad \sigma \times \sigma = 1+f\,.
  \end{equation}
  The one-form symmetry is $\Z/2$ generated by $f$.
 \end{itemize}
By gauging the fermion parity one-form symmetry in $\Spin(n)_1$ we can map them to $\SO(n)_1$ which are invertible fermionic field theories. The latter (super-modular) categories have central charge $\frac{n}{2} \mod 8$ and the property that $\SO(n)_1\boxtimes \SO(n')_1 \simeq \SO(n+n')_1$, and we will use them to represent the line operators for the symmetry $\fC$ which is a fermionic strongly fusion 2-category, and only consists of the vacuum and $\psi$. The interpretation of using $\SO(n)_1$ to represent the line operators goes back to Figure \ref{fig:twostep}. In particular the starting theory $\cT$ also can start off with a gravitational counter term (or gravitational anomaly), and that is related to the central charge of whichever $\SO(n)_1$ that we select to represent $\sVect$.

At the level of the (2+1)d theory $\mathcal{T}$ there is an obvious tensor product $\SO(n)_1 \boxtimes \Spin(m)_1$, which gives a fermionic theory. At the level of the symmetry i.e. strongly fusion 2-category, the analogous tensor product is given by $\Sigma \SO(n)_{1} \boxtimes \Sigma\Spin(m)_1$ which includes all the \text{condensations} in the sense of \cite{Gaiotto:2019xmp}. These are surface defects that are also present in $\mathcal{T}$ with the property that they can be ended on a line operator. In other words the surface defects are completely determined by the line defects.
We now condense the anyons in the algebra $\cA =1 \boxtimes 1 \oplus \psi \boxtimes f$, mathematically this amount to computing the  category of local modules of the algebra in $\SO(n)_1 \boxtimes \Spin(n)_1$. By condensing the anyons at the level of the theory, this will completely determine the structure of the symmetry.

\begin{lemma}
   Taking local modules of the algebra $\mathcal{A}=1 \boxtimes 1 \oplus \psi \boxtimes f$ in $\SO(n)_1 \boxtimes \Spin(m)_1$ gives $\SO(n+m)_1$.
\end{lemma}
\begin{proof}
The theory defined by  $\SO(n)_1 \boxtimes \Spin(m)_1$ is fermionic, and the algebra $\cA$ is bosonic. Hence the resulting theory after condensing this algebra must have a local fermion.
    We first compute the modules, and then the local modules following \cite{Yu:2021zmu}. The computations for the case when $n$ is even or add are similar, so we will present the odd case. The modules are given by 
    \begin{align}
        \mathcal{A} \otimes (1 \boxtimes 1 )&=  1\boxtimes 1 \oplus \psi \boxtimes f,\quad  \mathcal{A} \otimes (\psi \boxtimes 1) =  1 \boxtimes f \oplus \psi \boxtimes 1,  \\\notag 
        \mathcal{A} \otimes (1 \boxtimes f )&= \psi \boxtimes 1 \oplus 1 \boxtimes f ,\quad  \mathcal{A} \otimes (\psi \boxtimes f) = \psi \boxtimes f \oplus 1\boxtimes 1, \\ \notag 
         \mathcal{A} \otimes (1 \boxtimes \sigma )&= 1\boxtimes \sigma \oplus \psi \boxtimes \sigma ,\quad  \mathcal{A} \otimes (\psi \boxtimes \sigma) = \psi \boxtimes \sigma \oplus 1\boxtimes\sigma\,.
    \end{align}
    In the local modules, the algebra $\cA$ becomes the new vacuum line. The object $1\boxtimes f \oplus \psi \boxtimes 1$ comprises of objects that differ by an integer spin, and are thus included in the local modules where it become the local fermion. The object $\psi \boxtimes \sigma \oplus 1\boxtimes\sigma$ comprises of objects that differ $\frac{1}{2}$ in their spin, and thus not a local module.  
   The central charge is also preserved by Witt equivalence. Hence, in the result we must sum the ranks of the two inputs into the tensor product.
\end{proof}

\begin{theorem}\label{prop:shift}
    Let $\fC$ be a fermionic strongly fusion 2-category defined by a fermionic fusion 2-category defined by the bosonic group $G_b$, a class $\kappa \in \rH^2(BG_b;\Z/2)$, a supercohomology class $\varpi \in \SH^{4+\kappa}(BG_b)$ with nontrivial Majorana layer, and the data of a choice of $\SO(n)_1$ (for any $n \mod 16$) corresponding to $\sVect$. Applying Construction \ref{construction:F2Cmodify} on $\fC$ exhibits the following shift of $\varpi$ as a cocycle:
    \begin{itemize}
        \item if $\kappa=0$ then $\varpi$ has no shift,
        \item if $\kappa \neq 0$ and $\Sq^1\kappa =0$, then $\varpi$ has an order 2 shift,
        \item if $\kappa \neq 0$ and $\Sq^1\kappa \neq 0$, then $\varpi$ has an order 4 shift.
     \end{itemize}
\end{theorem}

\begin{proof}
    We first argue why $\alpha$ necessarily shifts to $\alpha+\kappa$ under one application of stacking and condensing the algebra $\mathcal{A}$. In $\fC$ the data of $\alpha$ is the extension of $\pi_0(\sVect)^\times \cong \Z/2$ by $G_b$ giving the group $G_f$, and $\kappa$ provides and action of $G_b$ on $\sVect$. We claim that upon stacking by $\Spin(n)_1$, the action of $G_b$ on $\Spin(n)_1$ must also be via the same class $\kappa$.
    In particular, $\kappa$ is how the $G_b$-action fractionalizes on the $\Z/2$ one-form symmetry generated by the fermion line in $\Spin(n)_1$. The fractionalization has the effect of changing the spin of the anyon $f$ \cite{Hsin:2019gvb}, and since our procedure requires us to condense $\psi \otimes f$, then to be sure that the composite is a boson, we must take the action on $\Spin(n)_1$ to be the same as the action on $\sVect$.

     Upon condensing the algebra given by $1\boxtimes 1\oplus \psi\boxtimes f$ it is known that the symmetry fractionalization data $\kappa$ becomes an extension for $G_b$ by the dual $\Z/2$ 0-form symmetry to the 1-form symmetry generated by $f$  \cite{Yu:2020twi,Tachikawa:2017gyf}. 
     %One can see this by noting that $\hom(B^2\Z/2;\mathbb{C}^\times)=\rH^2(B^2\Z/2;\mathbb{C}^\times) \cong \Z/2$ is the dual $\Z/2$ 0-form symmetry which  has to be identified with the $\Z/2$ coming from $\pi_0(\sVect)^\times$.
    Hence, the total extension data for $G_b$ by $\Z/2$ is $\alpha+\kappa$.  Thus, we see that $\alpha \mapsto \alpha+\kappa$ after one iteration of stacking and condensing, and each successive iteration adds on another copy of $\kappa$. Therefore $\alpha$ is 2-periodic if $\kappa$ is nontrivial.
    
   % In other words, it is a mixed anomaly between the $G_b$ symmetry and the $\Z/2$ one-form symmetry carried by the fermion line in $\Spin(n)_1$. 
    We now track how the shift of $\alpha$ propagates and affects the cochains $\beta$ and $\gamma$. Let $\beta' = \beta + \Delta \beta$, where $\Delta \beta$ labels the shift in $\beta$ after stacking and then condensing the algebra $\cA$. We will use Definition \ref{def:supercohclass}, which explains how $\beta$ is acted upon by differentials, to give the explicit form of $\Delta \beta$ in terms of $\alpha$ and $\kappa$. In particular we have
    \begin{align}
        d\beta' &=(\Sq^2+\kappa)(\alpha+\kappa)\\ \notag 
    &=(\alpha+\kappa)^2+\kappa\alpha+\kappa^2
    \\ \notag 
    &=\alpha^2+\kappa\alpha+\alpha\kappa+\kappa^2+\kappa\alpha+\kappa^2\\ \notag 
    &= \alpha^2+ \alpha\kappa,
    \end{align}
    where the multiplication between two cochains means applying $\cup$.
    Notice that $d\beta' = d\beta+d\Delta\beta$, hence $\Delta \beta$ takes a form such that $d\Delta \beta = \kappa \alpha + \alpha \kappa$. In particular, $\Delta \beta = \kappa \cup_1 \alpha$.\footnote{The expression for $\Delta \beta$ can be found universally, which means that there are no contributions from other closed forms.} Applying this procedure again would amount to letting $\beta'' = \beta + \Delta \beta + \Delta \beta'$, and finding the form of $\Delta \beta'$. In particular, $\Delta \beta' = \kappa \cup_1 (\alpha+\kappa)$, and therefore $\beta'' = \beta +\kappa \cup_1 \kappa = \beta + \Sq^1 \kappa$. If $\Sq^1\kappa$ is trivial, then $\beta$ is 2-periodic, and if $\Sq^1 \kappa$ is nontrivial, then $\beta$ is 4-periodic. The 4-periodicity appears essentially because iterating this process twice more will generate another copy of $\Sq^1\kappa$ that cancels the first and takes us back to $\beta$.
    
    For the cocycle $\gamma$, we note that in stacking and condensing twice we get
\begin{equation}\label{eq:gammadouble}
    d \gamma'' = i((\Sq^2 +\kappa)(\beta+\Sq^1\kappa))+f_\kappa(\alpha),
\end{equation}
 which is equivalent to $d \gamma$ if $\Sq^1 \kappa$ is trivial. Therefore $\gamma''$ and $\gamma$ differ by a term which is closed. In particular, let $\gamma'' = \gamma + \Delta \gamma + \Delta \gamma'$, where $\Delta \gamma + \Delta \gamma'$ is closed, and neither are themselves closed since by assumption each shift is a new cochain. But the only nontrivial degree 4 closed terms involving  $\cup$ or $\cup_1$ are $\alpha^2$, $\kappa\alpha$, $\alpha\kappa$, $\kappa^2$, and  linear combinations thereof. However none  can feed into the expression of $\gamma$ as they cannot be written as a the sum of two cochains. Therefore we conclude that $\Delta \gamma$ and $\Delta \gamma'$ cancel out in the sum and $\gamma$ is  2-periodic if $\Sq^1 \kappa$ is trivial.

 All that is left is to consider the periodicity of $\gamma$ in the case when $\Sq^1\kappa$ is nontrivial. We note that by Equation \ref{eq:gammadouble}, if we took $\beta \mapsto \beta+\Sq^1 \kappa$ and kept $\alpha$ invariant, then the right-hand-side would be equivalent to $d\gamma$, and  this implies that $d\gamma^{(4)} = d\gamma$.
 For shorthand, let $\Delta \Gamma = \Delta \gamma +\Delta \gamma'$, and $\Delta \Lambda = \Delta \gamma^{''} +\Delta \gamma^{'''}$. Then $\gamma^{(4)}$ takes the form $\gamma^{(4)}= \gamma + \Delta \Gamma+\Delta \Lambda$, where $\Delta \Gamma+\Delta \Lambda$ is closed. Then by the same logic used previously, we conclude that $\Delta \Gamma+\Delta \Lambda$ is trivial, and therefore $\gamma$ is 4-periodic if $\Sq^1\kappa$ is nontrivial.

\end{proof}

It is manifestly clear that applying Construction \ref{construction:F2Cmodify} can change the fusion 2-category and hence the symmetry of $\cT$. For example, if $\kappa=\alpha$ then the shift in $\alpha$ will trivialize the extension that defines $G_f$, leading to a different fusion structure of the symmetry operators.
A question now arises, concerning whether or not the theories $\cT$ and $\cT_{+\frac{n}{2}}$ are isomorphic if $n$ is not equivalent to 0  mod 16. If one claims that they are not isomorphic then having a different symmetry structure is perhaps not surprising. One could however claim that the two quantum field theories  are equivalent since a difference by the central charge cannot be distinguished without coupling to gravity or placing the theory on a manifold with boundary. This result is then a bit more startling as it would suggest that the same theory is compatible with multiple symmetries.

The resolution lies in allowing for the procedure of stacking by $\Spin(n)_1$ and condensing the boson to be part of the equivalence condition that defines a fermionic QFT $\mathcal{T}$. In light of there being nontrivial TQFTs in (2+1)-dimensions, it appears natural to generalize to this valent structure, that goes beyond the traditional framework of stacking with invertible TQFTs alone. See also \cite{Bhardwaj:2024xcx} for another example where nontrivial (2+1)d TQFTs makes an impact from the point of view of higher categorical symmetries involving self-duality. With this choice of valence, the full data of the symmetry would be all possible fusion 2-categories accessible via shifting $\varpi$ by the rules of Proposition \ref{prop:shift}. For this reason, we take as a slogan that the symmetries of a QFT should be a \textit{valence dependent notion} and in particular can be associated to a set of categories up to relations. In this sense, it is not unreasonable to expect that different physics may emerge when varying the valence structures, and we will see an examples of this in a following remarks.
The novel fact that we uncover is that in (2+1)-dimensions, even the underlying fusion rules for the symmetry can change. 
One could view the TQFT that was stacked on to have interacted in some way with the local degrees of freedom in $\cT$, hence changing their quantum numbers to be compatible with a new symmetry structure after condensing the theory away. 

To give a bit more physical context of valence, we provide the following two remarks.

\begin{rem}
Suppose $| \psi \rangle_d$ is a state for a $d$-spacetime-dimensional quantum system. One can consider the following equivalence:
    \begin{equation}\label{eq:valencestates}
        \mathcal{U}(| \psi \rangle_d \otimes |\text{triv}\rangle_{d-1}) \cong | \psi'\rangle_{d} \otimes | \text{triv}'\rangle_{d-1}.
    \end{equation}
    Here, the initial state $| \psi \rangle_d$ is said to be equivalent to the final state $| \psi' \rangle_d$  if it is possible to tensor $| \psi \rangle_d$  with  some locally trivial resource $|\text{triv}\rangle_{d-1}$,
    and evolve under the operator $\mathcal{U}$ to the state $| \psi' \rangle_d$, tensored also with a trivial resource. The resource state, or ancilla, $|\text{triv}\rangle_{d-1}$  is a state that is allowed to be inserted in a $d-1$-dimensional submanifold. The resource is local in the sense that it only occupies a submanifold rather than being allowed to extend to the entire $d$-manifold, and it is trivial in the sense that it can be connected to the trivial product state by some finite depth local unitary.  
    
    One could in principle alter the equivalence relation in \eqref{eq:valencestates} to be of the form
     \begin{equation}
        \mathcal{U}(| \psi \rangle_d \otimes |\phi \rangle_{d-1}) \cong | \psi'\rangle_{d}  \otimes |\phi' \rangle_{d-1}
    \end{equation}
   where the ancilla $|\phi \rangle_{d-1}$ need not be trivial, but could be an invertible state. This is analogous to how we propose a different valence and allow for stacking by nontrivial TQFTs, rather than only SPTs. In the context of (1+1)d quantum cellular automata that respects some global symmetry $G$, this latter type of equivalence resembles what is called a \textit{stable equivalence} rather than a strong equivalence.  Under a stable equivalence the ancilla state $|\phi \rangle$ is allowed to transform in an arbitrary representation of the $G$-symmetry, but a strong equivalence mandates that the ancilla  must transform in the same $G$-representation as the original state $| \psi \rangle$ \cite{Zhang:2023upl}. 
\end{rem}

\begin{rem}
There are exotic topological phases called fracton phases, first introduced in \cite{Haah:2011drr,haah2014bifurcation}. These phases can arise as topological defect networks. One views a  $d$-dimensional  TQFT as arising from introducing new interactions that are spatially localized on some lower $k<d$ dimensional region, which do not close the bulk gap \cite{Aasen:2020zru}.
In fractonic theories such as the X-cube model in (3+1)d, it is indeed useful to consider a different choice of equivalence under stacking by nontrivial TQFTs in (2+1)d. In some sense, fractons can be defined as those topological phases which arise when one looks for theories that are equivalent under this particular choice of valence. The interesting physical features that distinguishes them from the typical topological orders can also be attributed to the fact that they have this valence property, while more standard topological orders in the sense of \cite{Kong:2014qka,kong2015boundary,kong2017boundary,Lan_2018,Lan_2019,JF} are classified up to gapped boundary.
\end{rem}

\begin{rem}
    If we had not just taken the the $\kappa$ twist of supercohomology, but incorporated a degree 1 twist for the fusion 2-category by going along the lines of the general definition of twisted spin structure, then we would not see this phenomena. This is because if we had a  twisted orientation then the central charge would have to canonically be zero in order for the TQFT to be placable on unoriented manifolds. Furthermore, if we considered a fully bosonic QFT $\mathcal{T}$  with a bosonic strongly fusion 2-category as the symmetry, there is no such analog of Construction \ref{construction:F2Cmodify}.
\end{rem}

\begin{rem}\label{rem:fermionicTQFT}
   We will here make a distinction between the situations that use twisted supercohomology versus untwisted supercohomology. While categorical symmetries depends on classes in twisted supercohomology, fermionic TQFTs in (3+1)d which are associated to $\tsVect$-enriched braided fusion 2-categories and do not depend on the twisting $\kappa$. 
     More precisely we actually need to take $\kappa=0$ to define the TQFT. This is due to the fact that the category $\Omega \cZ(\tsVect^{\varpi}_{G_b}) = \Rep(G_b)\boxtimes \sVect$ if $\kappa=0$, for which we have a fully faithful functor $\sVect \to \Omega \cZ(\tsVect^{\varpi}_{G_b})$, which determines the trivial line and the fermion as the only objects that braid trivially with all other objects.
\end{rem}

%%%%%%%%%%%%%%%%%%%%%%%%%%%%%%%%%%%%%%%%%%%%%%%%%%
\section{Interactions of the Symmetry and TQFT}\label{section:TQFTandSym}
%%%%%%%%%%%%%%%%%%%%%%%%%%%%%%%%%%%%%%%%%%%%%%%%%%

In the last section we have seen how the structure of the fusion 2-category symmetry of a fermionic quantum field theory can be affected by stacking with TQFTs.
It is not surprising that the symmetry also affects which TQFTs are stackable onto the QFT, since the TQFTs are sensitive to the tangential structure of the background manifold that is set by the properties of the symmetry.

In this section we show that the conditions that are needed to exhibit a $\Z/2$ and $\Z/4$ periodicity of the cocycle $\varpi$ in Proposition \ref{prop:shift} exactly reflect the fact that some of the $\Spin(n)_1$ theories cannot be defined on manifolds with a twisted spin structure. By the discussion in \S\ref{section:fermionicsym} the twisted spin structure is set by the same conditions that give rise to the twisted supercohomology for which $\varpi$ is a cocycle. The exact same conditions also control the image of the map $F: \Mext(\cE) \rightarrow \Mext(\sVect)$; tying  together each of the three points will be the technical heart of Proposition \ref{prop:mainB}. 

The map $F$ arises from the following short exact sequences, the first is
\begin{equation}\label{eq:Mext}
    \begin{tikzcd}
        \Mext(\cE) \arrow[r] & \Witt \arrow[r]& \Witt(\cE)\,.
    \end{tikzcd}
\end{equation}
Here, $\Witt$ is the Witt group of non-degenerate braided fusion categories \cite{davydov2013witt}, and $\Witt(\cE)$ is the group of non-degenerate braided fusion categories over $\cE$ i.e. elements are a Witt class $[\cC]$ for $\cC$ a braided fusion category with $\cZ_{(2)}(\cC)\cong \cE$, and composition with respect to another class $[\cD]$ is given by $[\cC\boxtimes_\cE \cD]$ up to $\cE$-Witt equivalence. The second map in Equation \eqref{eq:Mext} is given by  base change, sending $[\cA] \mapsto [\cA \boxtimes \cE]$. In the case of $\cE = \sVect$ we have a sequence
\begin{equation}\label{eq:MextsVect}
    \begin{tikzcd}
        \Mext(\sVect) \arrow[r] & \Witt \arrow[r]& \Witt(\sVect)\,,
    \end{tikzcd}
\end{equation}
and due to the fact that there is a map  $(-)\boxtimes_\cE \sVect:\Witt(\cE)\rightarrow \Witt(\sVect)$, the map $\Mext(\cE) \to \mathcal{W}itt$ factors through  $\Mext(\sVect)$. If $\cE$ is super-Tannakian, it contains a maximal Tannakian subcategory $\cF$. If $\cC$ is a minimal extension of $\cE$ then the de-equivariantization by $\cC^0 = \cC \boxtimes_{\cF} \mathbf{Vect}$  is a minimal extension of $\sVect$.
The map $F:\Mext(\cE)\to \Mext(\sVect)$ assigns $\cC \mapsto \cC^0$, and is known as the \textit{central charge map}.

%%%%%%%%%%%%%%%%%%%%%%%%%%%%%%%%%%%%%%%%%%%%%%%%%%
\subsection{The (3+1)d TQFT for Fusion 2-Categories}
%%%%%%%%%%%%%%%%%%%%%%%%%%%%%%%%%%%%%%%%%%%%%%%%%%
In this section, we aim to build physical intuition for the map $F$ by unpacking its structure within the framework of a (3+1)-dimensional TQFT. Along the way, we revisit how fusion 2-categories relate to problems in 1-categories, shedding light on why shifts of the cocycle are relevant to the study of minimal non-degenerate extensions of braided fusion 1-categories.

For a braided fusion 1-category $\cB$ with Müger center $\cE$, a MNE of $\cB$ is a non-degenerate braided fusion category $\mathcal{C}$ such that $\cB \hookrightarrow \cC$. Furthermore, we require that the centralizer of $\cB$ in $\mathcal{C}$ i.e. objects which braid trivially with $\cB$ in $\cC$ are given by $\cE$. 

\begin{rem}
    A physical way of framing MNEs can be done from the lens of (2+1)d TQFTs. A non-degenerate braided fusion 1-category gives rise to what is known as a pure state TQFT. Such a theory has no influence from the environment, and therefore not subject to local decoherenece.
    On the other hand a braided fusion 1-category with nontrivial Müger center will give rise to a mixed state TQFT \cite{Sohal:2024qvq,Ellison:2024svg}.  These theories occur in the presence of interactions with background that disturb the pure state features. Being able to find an MNE for a degenerate braided category, means being able to find a \textit{purification} for the mixed state topological order.
\end{rem}

In the work of \cite{johnson2024minimal}, it was shown that a minimal non-degenerate extension of $\cB$ with Müger center $\sVect$ always exists. A large portion of the proof required an understanding of (3+1)d TQFTs, and how they can be constructed from the Drinfeld center of a fusion 2-category. We apply the methods of \cite{johnson2024minimal} to gain a better physical understanding of the last two points of Theorem \ref{prop:mainB} which concerns the map $F$ and the twisted spin structures.

Consider a  braided fusion 1-category $\cB$  with $\cZ_{(2)}(\cB)\cong \cE$, a Tannakian category of the form $\cE = \Rep(G_b)$.
In the proof of the existence of minimal non-degenerate extensions  one starts by establishing an isomorphism $\cZ(\Sigma \cB)\cong \cZ(\Sigma \cE)$. With this, one can form the non-degenerate braided category $\cC$, with $\cB \hookrightarrow \cC \hookleftarrow \cE^{rev}$, i.e. that $\cB$ and $\cE$ are centralizing pairs in $\cC$. The category $\cZ(\Sigma \cB)$ can be used to construct a (3+1)d TQFT which can be simplified by the following procedure that classifies the TQFT up to gapped boundary i.e. the category up to Morita equivalence: 
\begin{itemize}
    \item The lines in the TQFT are symmetric monoidal and given by $\cE$.  We can apply the fiber functor $\Rep(G_b)\rightarrow \mathbf{Vect}$. This creates a gapped boundary with a $G_b$-symmetry.
    \item By nondegeneracy, the TQFT must not have any genuine surface operators. This means there is a gapped boundary from the theory given by $\cZ(\Sigma \cB)$ to the vacuum. Surface operator can arise but only as condensation operators.
    \item $\cZ(\Sigma\cB)$ can be constructed by gauging the $G_b$-symmetry on the vacuum i.e. gauging the $G_b$-SPT. Hence $\cZ(\Sigma\cB)$ is a $G_b$-gauge theory, i.e. a Dijkgraaf-Witten theory with action $\pi \in \rH^4(BG_b;\bC^\times)$. 
\end{itemize}
If $\cZ(\Sigma \cB)$ is to be isomophic to $\cZ(\Sigma \cE)$ then the Dijkgraaf-Witten theory with action $\pi$ must be trivial, and therefore $\pi$ is the obstruction to a MNE. The choices of isomorphism for $\cZ(\Sigma \cB)\cong \cZ(\Sigma \cE )$ corresponds to the choices of minimal non-degenerate extension for $\cB$. 
Given a symmetric  fusion category $\cE$ there is a group $\Mext(\cE)$ defined to be the group of minimal non-degenerate extensions of $\cE$ \cite{lan2017modular}. For example, it is known that $\Mext(\sVect) \cong \Z/16$ \cite{bruillard2017fermionic}, giving Kitaev's 16-fold way \cite{Kitaev:2005hzj}, and other examples of this group have been computed for $\cE = \Rep(G_f)$ in \cite{nikshych2022computing}. 

We move on to consider when $\cE$ is just $\sVect$, then the theory constructed from $\cZ(\Sigma \sVect)\cong \cZ(\tsVect)$ is the theory of a dynamical spin structure. This category has two components, which we call the identity component and the \textit{magnetic} component. In the identity component, the only line operators are given by the vacuum and the emergent fermion. The simple nontrivial object is given by the surface operator $c$ which implements a  $\Z/2$ 1-form symmetry and arises as a condensation of the emergent fermion.  The magnetic component exists due to  non-degeneracy of the category. In this component there is a simple object $m$ that links nontrivially with the emergent fermion. 

We can realize $\tsVect$ as a topological boundary of $\cZ(\tsVect)$ by taking the Neumann boundary condition, which allows the spin structure to vary on the boundary.  For $\cD \in \Mext(\sVect)$ using the Morita equivalence $\tsVect \boxtimes \Sigma \cD \xrightarrow{\sim} \tsVect$, 
we get the following isomorphism on centers
\begin{equation}\label{eq:isocenter}
   \cZ(\tsVect \boxtimes \Sigma \cD ) \cong \cZ(\tsVect)\,.
\end{equation}
We view $\cZ(\tsVect)$ as an object in the proper Morita 5-category of braided fusion 2-categories, constructed following \cite[Section 8]{JFS}. The 1-morphisms of such a 5-category are given by monoidal (Karoubi-complete) $\mathbb{C}$-linear 2-category bimodules.
The product $\tsVect \boxtimes \Sigma \cD$ is a boundary for $\cZ(\tsVect \boxtimes \Sigma \cD )$, and by the isomorphism above, can be interpreted as a boundary for $\cZ(\tsVect)$ formed by stacking $\Sigma \cD$ on the Neumann boundary condition. Thus  $\tsVect \boxtimes \Sigma \cD$ is a bimodule and is furthermore an invertible 1-morphism of the 5-category. We see that stacking on $\Sigma \cD$ gives rise to an automorphism of $\cZ(\tsVect)$.  Automorphisms of $\cZ(\tsVect)$ where also studied in \cite{Barkeshli:2023bta} in the context of quantum computing. Pictorially this result is depicted in Figure \ref{fig:BosSandwich}.
\begin{figure}[ht]
\centering
 \begin{tikzpicture}[thick]
    
        % Dimensions
        \def\Depth{4}
        \def\DepthTwo{3}
        \def\Height{2}
        \def\Width{2}
        \def\Sep{3}        
        
        % 3d Manifold on double segment
        \coordinate (O) at (0.5,0,0);
        \coordinate (A) at (0.5,\Width,0);
        \coordinate (B) at (0.5,\Width,\Height);
        \coordinate (C) at (0.5,0,\Height);
        \coordinate (D) at (\Depth-.5,0,0);
        \coordinate (E) at (\Depth-.5,\Width,0);
        \coordinate (F) at (\Depth-.5,\Width,\Height);
        \coordinate (G) at (\Depth-.5,0,\Height);
        \draw[black] (O) -- (C) -- (G) -- (D) -- cycle;% Bottom Face
        \draw[black] (O) -- (A) -- (E) -- (D) -- cycle;% Back Face
        \draw[black, fill=white!60,opacity=0.9] (O) -- (A) -- (B) -- (C) -- cycle;% Left Face
        \draw[black, fill=yellow!20,opacity=0.8] (D) -- (E) -- (F) -- (G) -- cycle;% Right Face
        \draw[black] (C) -- (B) -- (F) -- (G) -- cycle;% Front Face
        \draw[black] (A) -- (B) -- (F) -- (E) -- cycle;% Top Face
        \draw[black, fill=green!20,opacity=0.8] (\Depth,0,0) -- (\Depth,\Width,0) -- (\Depth,\Width,\Height) -- (\Depth,0,\Height) -- cycle;
       % \draw[below] (0+1, 0*\Width, \Height) node{$\Sigma \cB$};
        \draw[below] (\Depth, 0*\Width, \Height) ;
        %node{$\cE$};
        \draw[midway] (\Depth/2,\Width-\Width/2,\Height/2) node {$\cZ(\tsVect)$};
        \draw[midway] (\Depth/2+3.5,\Width-\Width/2,\Height/2) node {$\cZ(\tsVect)$ };

         \draw[midway] (\Depth/2+.2,\Width-3.5,\Height/2-2) node {$\tsVect$ };

         \draw[midway] (\Depth/2+.2,\Width-1.5,\Height/2-6) node {$\Sigma \cD$ };

        \coordinate (O2) at (\Depth-.5,0,0);
        \coordinate (A2) at (\Depth-.5,\Width,0);
        \coordinate (B2) at (\Depth-.5,\Width,\Height);
        \coordinate (C2) at (\Depth-.5,0,\Height);
        \coordinate (D2) at (\Depth+\DepthTwo,0,0);
        \coordinate (E2) at (\Depth+\DepthTwo,\Width,0);
        \coordinate (F2) at (\Depth+\DepthTwo,\Width,\Height);
        \coordinate (G2) at (\Depth+\DepthTwo,0,\Height);
        \draw[black] (O2) -- (D2);
        \draw[black] (A2) -- (E2);
        \draw[black] (B2) -- (F2);
        \draw[black] (C2) -- (G2);

        \draw[black, fill=white!20,opacity=0.8] (\Depth+3,0,0) -- (\Depth+3,\Width,0) -- (\Depth+3,\Width,\Height) -- (\Depth+3,0,\Height) -- cycle;
        % \draw[below] (\Depth+3, 0*\Width, \Height) node{$\tsVect$};
       
    \end{tikzpicture}
 \caption{The yellow square denotes the Neumann boundary condition, giving rise the the category $\tsVect$ on the boundary. The green square denotes $\Sigma \cD$, and stacking is implemented by zooming out on the middle configuration. Together, the composite gives rise to a bimodule from the center to itself using Equation \eqref{eq:isocenter}.}
    \label{fig:BosSandwich}
\end{figure}
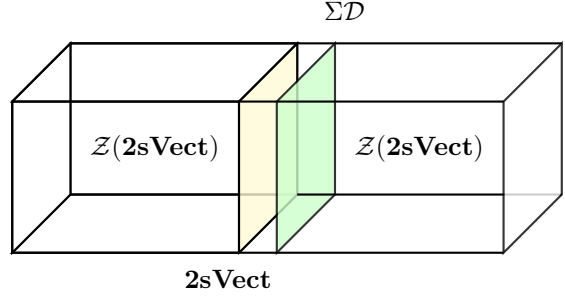

We turn to the case where $\cZ_{{2}}(\cB)\cong \cE$ is super-Tannakian. Categories $\cE$ is classified by fermionic groups, i.e.  $G_b$ together with the extension data $\kappa \in \rH^2(BG_b;\Z/2)$ that classifies the extension.  In the case when the extension is trivial, we use the fiber functor $\Rep(G_b)\rightarrow \mathbf{Vect}$ on the bosonic part, leaving just $\sVect$ in the Müger center. In the case when the extension is not trivial, then we use the super fiber functor $\Rep(G_f)\rightarrow \sVect$ on the Müger center.\footnote{It is expected that a similar argument that involves super fiber functors for $\mathbf{nRep}(G_f)$ can help classify higher dimensional topological orders, though the existence of such fiber functors has not been made rigorous. The existence of a 2-fiber functor for symmetric fusion 2-categories has been proven in \cite{Decoppet:2022dnz}.} In both cases, we obtain a gapped interface, as shown in  Figure \ref{fig:composeSandwich}, to $\cZ(\tsVect)$.
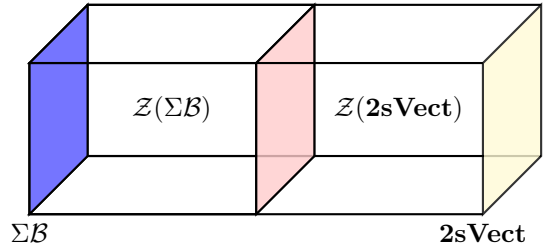
\begin{figure}[ht]
\centering
    \begin{tikzpicture}[thick]
    
        % Dimensions
        \def\Depth{4}
        \def\DepthTwo{3}
        \def\Height{2}
        \def\Width{2}
        \def\Sep{3}        
        
        % 3d Manifold on double segment
        \coordinate (O) at (0+1,0,0);
        \coordinate (A) at (0+1,\Width,0);
        \coordinate (B) at (0+1,\Width,\Height);
        \coordinate (C) at (0+1,0,\Height);
        \coordinate (D) at (\Depth,0,0);
        \coordinate (E) at (\Depth,\Width,0);
        \coordinate (F) at (\Depth,\Width,\Height);
        \coordinate (G) at (\Depth,0,\Height);
        \draw[black] (O) -- (C) -- (G) -- (D) -- cycle;% Bottom Face
        \draw[black] (O) -- (A) -- (E) -- (D) -- cycle;% Back Face
        \draw[black, fill=blue!60,opacity=0.9] (O) -- (A) -- (B) -- (C) -- cycle;% Left Face
        \draw[black, fill=red!20,opacity=0.8] (D) -- (E) -- (F) -- (G) -- cycle;% Right Face
        \draw[black] (C) -- (B) -- (F) -- (G) -- cycle;% Front Face
        \draw[black] (A) -- (B) -- (F) -- (E) -- cycle;% Top Face
        \draw[below] (0+1, 0*\Width, \Height) node{$\Sigma \cB$};
        \draw[below] (\Depth, 0*\Width, \Height) ;
        %node{$\cE$};
        \draw[midway] (\Depth/2+.5,\Width-\Width/2,\Height/2) node {$\cZ(\Sigma \cB)$};
        \draw[midway] (\Depth/2+3.5,\Width-\Width/2,\Height/2) node {$\cZ(\tsVect)$ };
        
        \coordinate (O2) at (\Depth,0,0);
        \coordinate (A2) at (\Depth,\Width,0);
        \coordinate (B2) at (\Depth,\Width,\Height);
        \coordinate (C2) at (\Depth,0,\Height);
        \coordinate (D2) at (\Depth+\DepthTwo,0,0);
        \coordinate (E2) at (\Depth+\DepthTwo,\Width,0);
        \coordinate (F2) at (\Depth+\DepthTwo,\Width,\Height);
        \coordinate (G2) at (\Depth+\DepthTwo,0,\Height);
        \draw[black] (O2) -- (D2);
        \draw[black] (A2) -- (E2);
        \draw[black] (B2) -- (F2);
        \draw[black] (C2) -- (G2);

        \draw[black, fill=yellow!20,opacity=0.8] (\Depth+3,0,0) -- (\Depth+3,\Width,0) -- (\Depth+3,\Width,\Height) -- (\Depth+3,0,\Height) -- cycle;
         \draw[below] (\Depth+3, 0*\Width, \Height) node{$\tsVect$};
       
    \end{tikzpicture}
 \caption{From the point of view of $\cZ(\tsVect)$, one can recover the TQFT given by $\cZ(\Sigma\cB)\cong \cZ(\tsVect^\varpi_{G_b})$  by gauging the $G_b$ or $G_f$-symmetry on the interface represented in red.}
    \label{fig:composeSandwich}
\end{figure}
One can view the components of $\cZ(\Sigma \cB) \cong \cZ(\tsVect^\varpi_{G_b})$ in this case as the extension of the 2 components arising from $\sVect$ with the components labeled by $G_b$. In this sense fermionic strongly fusion 2-category of the form  $\tsVect^\varpi_{G_b}$ are related to MNEs of $\cB$ with Müger center $\cE = \Rep(G_f)$, by looking at the center $\cZ(\tsVect^\varpi_{G_b})$.  In a similar manner, $\cZ(\tsVect)$ is related to MNEs of those braided categories with Müger center $\sVect$. One  can view $\cZ(\tsVect^\varpi_G)$ as the so called ``SymTFT'' for the fermionic strongly fusion 2-category symmetry $\tsVect^\varpi_{G_b}$.  In a similar way as the $\tsVect$ case, in  \cite[Remark 2.3.9]{Johnson-Freyd:2020twl} the author discusses that $\pi_0$ of the groupoid of automorphisms of $\cZ(\Sigma\cB)$ is  given by the group $\Mext(\cE)$, which was first conjectured in \cite{Kong:2020jne}, and the automorphisms are implemented by stacking with elements $\Sigma\cM$ for $\cM \in \Mext(\cE)$. We can think of stacking by these categories as overlaying on the red interface in Figure \ref{fig:composeSandwich}, similar to the configuration in Figure \ref{fig:BosSandwich}. In Figure \ref{fig:composeSandwich} we stack the right boundary of the theory $\cZ(\Sigma\cB )$ by $\Sigma \cM$, and apply the same procedure that we used to go from $\cZ(\Sigma\cB )$ to $\cZ(\tsVect)$. This is analogous to applying the map $F$ in the bulk, which on the boundary will map $\Sigma \cM$ to $\Sigma \Spin(n)_1$. This will not hit all of the $\Spin(n)_1$ theories, but we do get that:
\begin{corollary}\label{cor:autos}
   Let $\cB$ be a braided fusion 1-category with $\cZ_{(2)}(\cB) \cong \cE$ where $\cE$ is super-Tannakian. The $\Spin(n)_1$ theories that are in the image of the map $F:\Mext(\cE) \to \Mext(\sVect)$ arise from automorphism of 
     $\cZ(\Sigma \cB)$.
\end{corollary}
\noindent The automorphisms of $\cZ(\Sigma \cB)$ that in particular fix also the cocycle $\varpi$ associated to the symmetry on the boundary, have image in $\Mext(\sVect)$. This establishes the equivalence between \ref{point1} and \ref{point2} in Theorem \ref{prop:mainB}.  We will more rigorously show how to detect if a category in $\Mext(\sVect)$ can be coupled to the spin$\text{-}G_f$ structure in the next section, and therefore produce a map in the opposite direction of $F$ to go from a boundary of $\cZ(\tsVect)$ to a boundary of $\cZ(\tsVect^\varpi_{G_b})$. This will complete the proof of Theorem \ref{prop:mainB}. Intuitively, the symmetry boundary of $\cZ(\Sigma \cB)$ is given by $\tsVect^\varpi_{G_b}$ and therefore those $\Spin(n)_1$ in $\Mext(\sVect)$ which are in the image of $F$ are the theories that can consistently be coupled to the spin$\text{-}G_f$ structure, set by $\tsVect^\varpi_{G_b}$ and hence do not shift $\varpi$.

From the SymTFT point of view, even though the boundary symmetry of $
\cZ(\Sigma \cB)$ can change by stacking with $\Spin(n)_1$ and condensing, we claim that the SymTFT does not change. Unlike in Remark \ref{rem:fermionicTQFT}, the SymTFT is not strictly a fermionic theory and thus $\kappa$ is not necessarily trivial. Rather, it is constructed from a non-degenerate braided fusion 2-category. One way to approach the fact that the SymTFT is left invariant under the boundary manipulations is by observing that the Lagrangian description of the TQFT, viewed as a class in $\SH^{4+\kappa}(BG_b)$, should only be defined up to a shift of the Majorana layer \cite{Johnson-Freyd:2022}. Therefore it is not canonical. There is a long exact sequence 
\begin{equation}
    \begin{tikzcd}
       \ldots \arrow[r]&\mathrm{rSH}^{*}(BG_b) \arrow[r] & \SH^{*} (BG_b) \arrow[r,"M"] & \rH^{*-2}(BG_b;\Z/2) \arrow[r] &\ldots
    \end{tikzcd}
\end{equation}
where the map $M:\SH^{4+\kappa}(BG_b) \to \rH^2(BG_b;\Z/2)$  extracts the Majorana layer for the TQFT. The Majorana layer for the TQFT is only determined by the group $\rH^2(BG_b;\Z/2)/ \langle \kappa \rangle$, and therefore the TQFT is invariant upon stacking by $\Spin(n)_1$ on the boundary. As a result, the other two layers in the TQFT are unaffected.  In other words, possibly different spin$\text{-}G_f$ gauge theories, which depend on $\rH^2(BG_b;\Z/2)$, can give rise to the same TQFT.

%%%%%%%%%%%%%%%%%%%%%%%%%%%%%%%%%%%%%%%%%%%%%%%%%%
\subsection{Applications to Minimal Non-Degenerate Extensions}\label{subsection:applicationstoMNE}
%%%%%%%%%%%%%%%%%%%%%%%%%%%%%%%%%%%%%%%%%%%%%%%%%%
%\matt{We need to explain how the periodicity of the cocycle is to related to the image of the map $\Mext(\cE) \rightarrow \Mext(\sVect)$. Explain even what an MNE means. The motivation comes from the bulk theories to which $\Sigma \cE$ and $\tsVect$ are on the boundaries. The gauge theory in the latter case is given by .... , and in the former case is given by the center of something strongly fusion + dynamical spin structure.  }

%\matt{we could actually move this up. Have it as an extended remark to match with what is known in the literature.}
We have seen in the last section how the map $F:\Mext(\cE)\to\Mext(\sVect)$ arises in a physical way. It is instructive to connect our results for how $\varpi$ shifts to what is known about the map $F$ in \cite{nikshych2022computing}, where for certain $\cE$ some explicit examples of the subgroup in $\Mext(\sVect)$ have been computed. In particular we will connect with the following two Lemmas:
\begin{lemma}\cite[Proposition 5.4, Corollary 5.10]{nikshych2022computing}
    For $G_f =\Z/4$ and $G_f = \Z/2 \times \Z/2$ with $\cE = \Rep(G_f)$, then $\Mext(\cE)$ is isomorphic to $\Z/8$ and $\Z/16\times \Z/8$ respectively.
 \end{lemma}

\begin{lemma}\label{lemma:dmitri}\cite[Corollary 5.5]{nikshych2022computing}
   The image of $F:\Mext(\cE)\rightarrow \Mext(\sVect)$ is given by $\Z/8$ and $\Z/16$ for $\cE = \Rep(\Z/4)$ and $\cE =\Rep( \Z/2 \times \Z/2)$ respectively.
\end{lemma}

\begin{example}
    Thanks to Lemma \ref{lemma:dmitri} if $G_b = \Z/2$ then in $\rH^*(B\Z/2;\Z/2) = \Z/2[x]$ where $|x|=1$, $\kappa =  x^2$. However this means that $\Sq^1\kappa=0$ and therefore the periodicity of $\varpi$ is 2 by Theorem \ref{prop:mainA}, which means that of the 16 possible $\Spin(n)_1$ theories, only 8 leave the cocycle invariant.
\end{example}

\begin{example}
    Taking $G_b=\Z/2$ with $\kappa=0$ implies $G_f$ is split. The image of $F$ is $\Z/16$ by Lemma \ref{lemma:dmitri} and indeed  all $\Spin(n)_1$ leave the cocycle invariant in this case.
\end{example}

\begin{rem}
    In addition, the image of $F$ has been studied for example in  \cite{galindo2017categorical} where the authors found that the image is $\Z/16$ when $G_f$ is a trivial extension.
\end{rem}

We now turn to the proof of our second main proposition, which provides physical insight into the various cases of the shift in $\varpi$ described in Proposition \ref{prop:shift}:
\begin{theorem}\label{prop:TFAE}
 Let $\fC$ be a fermionic strongly fusion 2-category parametrized by the following pieces of data. A finite group $G_b$, a class $\kappa \in \rH^2(BG_b;\Z/2)$ such that $G_f = G_b\,  \textbf{.}_{\kappa} \Z/2$, a class $\varpi \in \SH^{4+\kappa}(BG_b)$. Let $\cE = \Rep(G_f)$, then the following are equivalent:
 \begin{enumerate}
     \item  \label{thm:point1} The $\Spin(n)_1$ theories such that application of Construction \ref{construction:F2Cmodify} on $\fC$ stacked with $\Sigma \Spin(n)_1$ leaves $\varpi$ invariant.
     \item \label{them:point2} The image of $F:\Mext(\cE) \to \Mext(\sVect)$.
     \item \label{thm:point3} The $\Spin(n)_1$ that can be coupled to background spin$\text{-}G_f$ structure set by the symmetry $\fC$.
 \end{enumerate}
\end{theorem}
Before proving this, we first need to have a better understanding of which of the theories $\Spin(n)_1$ can be coupled to background spin$\text{-}G_f$ structure. This can be gleaned from the work of \cite{Barkeshli:2021ypb}, where the authors give the classification of $G_f$-symmetric invertible fermionic phases. In order to do the classification the authors bosonized the invertible fermionic theories by gauging the diagonal $\Z/2$ fermion parity symmetry in $\Spin_3 \times G_f$. This in effect sends an invertible fermionic topological theory to one the $\Spin(n)_1$ bosonic theories which is classified by the central charge $c$. The authors then studied the classification of the $G_b$-symmetry enrichment on top of these bosonic topological theories.

\begin{ansatz}\label{ansatz:invertiblespin}\cite{Barkeshli:2021ypb}
    Invertible fermionic phases are parametrized by a  quadruplet denoted by $(c,n_1,n_2,\nu_3) \in \frac{1}{2}\Z \times \rC^1(BG_b;\Z/2) \times \rC^2(BG_b;\Z/2)\times \rC^3(BG_b;\mathbb{C}/\Z)$, along with the following consistency equations that each of the cochains must satisfy:
\begin{align}
    dn_1&=0\,, \label{eq:consistent1}\\
    dn_2 &= n_1 w_2 + c \,\Sq^1 w_2 \mod 2, \label{eq:consistent2}\\
    d\nu_3 &= \mathcal{O}_4[c,n_1,n_2] \label{eq:consistent3}\,,
\end{align}
where $\mathcal{O}_4$ is a $\mathbb{C}^\times$ valued cochain thought of as a 't Hooft anomaly.\footnote{The authors of \cite{Barkeshli:2021ypb} also provide the classification for an additional anti-unitary symmetry.}
\end{ansatz}

The data $(c,n_1,n_2,\nu_3)$ is related to how various defects of different codimension can be decorated with lower dimensional fermionic states, and also encodes the braiding and fusion properties of the defects. For each choice of $c$, $(n_1,n_2,\nu_3)$ form a torsor over a group
extension involving subgroups of $\rH^1(BG_b,\Z/2)$,  $\rH^2(BG_b,\Z/2)$, and  $\rH^3(BG_b,\mathbb{C}/\Z)$. The choice of $n_1$ determines (up to gauge transformations) how each element of $G_b$ acts on the fusion and splitting spaces of objects in $\Spin(n)_1$. In the presence of the $G_b$ symmetry, the anyons carry fractional $G_b$ quantum numbers. This is encapsulated in phases $\eta_{a}(g,h)$ with respect to each anyon $a$ for all $g,h\in G_b$. The choice of $n_2$ concerns how to define $\eta_a(g,h)$ with respect to a canonical reference value so that $\eta_f(g,h) = (-1)^{w_2(g,h)}$, and changing $n_2$ corresponds to changing the outcome of fusing
$g$ and $h$ defects by $f$ due to the fractionalization.\footnote{There is a redundancy in defining $n_2$, given by $n_2\simeq n_2+ w_2$, see \cite[Section C]{Barkeshli:2021ypb}. This redundancy takes the same form as the shift of $\alpha$ by $\kappa$.} 
The class $w_2$ is defined with respect to the background manifold, but due to the presence of the $G_b$-bundle, is given by a class in $\rH^2(BG_b;\Z/2)$. In terms of the data used to describe fermionic strongly fusion 2-categories, this corresponds to $\kappa$. The choice $\nu_3$ concerns the $F$-symbols in the $G_b$-enriched theory. In order to make the $F$-symbols satisfy the pentagon equations exactly, even under changing the symmetry fractionalization class we need to supply the data of $\nu_3$ as a correction term. 
In conclusion, the classification of invertible phases (in the unitary case) has a dependency on the tangential structure twisted by the bundle $G_b$.

We will focus on the equation concerning $dn_2$, as it is the most crucial for probing whether or not a particular $\Spin(n)_1$ theory can be placed on a manifold with a certain tangential structure. If there is  an obstruction to being able to choose $n_1$ and $n_2$ that solves this equation, then that would be sufficient to rule this theory out as being consistent on the background. The specific form of $\mathcal{O}_4$ will not concern us as an inability to solve the last equation is not an inconsistency of the theory, but reflects the fact that the theory must live on the boundary for a bulk  (3+1)d theory. In \cite[Table 1]{Barkeshli:2021ypb} the authors work out how the values of $n$ put constraints on $w_2$.
\begin{lemma}\label{lem:oddc}
For odd $n$ the only consistent cocycle for $w_2$ in the equations of Ansatz \ref{ansatz:invertiblespin} is the trivial one.
\end{lemma}
\begin{proof}
The class $w_2$ is equal to $\kappa$ by the discussion above. 
So, by the proof of Proposition \ref{prop:shift}, it reflects the symmetry fractionalization of $f$ with respect to $G_b$. However by the operator content in \eqref{eq:Nodddata} the fractionalization is not compatible with the fact that $\sigma \times f = \sigma$   by \eqref{eq:fusionodd}.
\end{proof}

\begin{proof}[Proof of Proposition \ref{prop:TFAE}]
We will show that \ref{point1} and \ref{point3} are equivalent. Using the equivalence of \ref{point1} and \ref{point2} established in Corollary \ref{cor:autos} then completes the proof of the Proposition.
Clearly when $c\neq 0$ but $w_2=0$ and so $\kappa=0$, then we can trivially find $n_1$ and $n_2$ that solve Equations \eqref{eq:consistent1} and \eqref{eq:consistent2}. Therefore each of the 16 theories are consistent, and this agrees with the first bullet of Proposition \ref{prop:shift}. 

We now consider the case when $c$ is half integral. Then according to Lemma \ref{lem:oddc} these theories do not couple to a $G_b$-twisted spin structure. Let $c$ be integral so that $n$ is even but $\Sq^1w_2=0$. Then Equation \eqref{eq:consistent2} becomes  $dn_2 = n_1 w_2$ and can always be solved by taking $n_1$ to be trivial. This also matches with the fact that in Proposition \ref{prop:shift}, $\Spin(2n)_1$ preserves $\varpi$, while the theories of odd rank shift $\varpi$. Let $w_2 =\kappa \neq 0$ and $\Sq^1w_2\neq 0$, if $c$ is even then we can always solve Equation \eqref{eq:consistent2} since the term with $\Sq^1 w_2$ drops out, thus $\Spin(n)_1$ where $n \equiv 0 \mod 4$ are all consistent.  Therefore we are only left to consider the case when $c$ is odd, i.e. $n \equiv 2 \mod 4$. The ability to solve Equation \eqref{eq:consistent2} depends on the group $G_b$, and its $\Z/2$ cohomology, but there will be groups for which we cannot solve this equation. Any instance where there is a group $G_b$ with a degree 2 cohomology class $\kappa$ that twists $w_2(TM)$, and $\Sq^1 \kappa$ includes a nontrivial generator for the degree 3 cohomology of $G_b$, provides a sufficient obstruction. For example, let $G_b=S_4$ which has $\Z/2$ cohomology ring given by \cite{Adem}
\begin{equation}
    H(BS_4;\Z/2) = \Z/2[x,y,w]/\{xw=0\},\quad |x|=1, |y|=2, |w|=3\,.
\end{equation}
 If we consider a $(BS_4,0,y)$-twisted spin structure as in \S\ref{subsection:prelim}, then using the fact that $\Sq^1 y = x y+ w$ and plugging this into Equation \eqref{eq:consistent2} gives 
 \begin{equation}
     dn_2 = n_1 y + xy+w \mod 2
 \end{equation}
 and so the terms on the right must be cohomologically trivial. Since the generators $x$, $y$, and $w$ are not cohomologically trivial then we must choose $n_1$ such that the right hand side cancels out mod 2. However, this is impossible and so this consistency equation is never satisfied. We conclude that in the case where $\kappa\neq 0$ and $\Sq^1 \kappa \neq 0$ that only those $n =  0 \mod 4$ are consistent theories, agreeing with invariance of $\varpi$ in this case.
\end{proof}

%\begin{itemize}
 %   \item Remark on the implications of the SymTFT, which will be a spin$\text{-}G_f$ gauge theory.
  %   \item Remark about the relation to the fact that there is no twist of $w_1$ involved in the fusion 2-category and hence no time reversal involved. In this case you really can have different values of $c$. In the case when you have time-reversal, and you need to put your theory on an unorientable manifold, there is no choice but to have $c=0$. Our twist is by $\kappa$ which does not touch $w_1$ at all, so is oriented, and hence the central charge can be anything.
     
    % In the bosonic case, we claim there is no central charge issues and indeed you can see that the only twist you can do in that case is for the orientation, and this would reinforce that $c=0$ is the only consistent choice.
%\end{itemize}

% The $\Spin(n)_1$ that survive when coupled to spin$G_f$ structure are exactly the ones you stack with to obtain the periodicity in $\varpi$, and these are exactly the ones in the image of the map given by Dmitri.

%We give interpretation of these results that match with the consistence condition for invertible fermionic theories in \cite{Barkeshli:2021ypb}.

\bibliographystyle{alpha}
\bibliography{ref}

@article{Aasen:2020zru,
    author = "Aasen, David and Bulmash, Daniel and Prem, Abhinav and Slagle, Kevin and Williamson, Dominic J.",
    title = "{Topological Defect Networks for Fractons of all Types}",
    eprint = "2002.05166",
    archivePrefix = "arXiv",
    primaryClass = "cond-mat.str-el",
    doi = "10.1103/PhysRevResearch.2.043165",
    journal = "Phys. Rev. Res.",
    volume = "2",
    pages = "043165",
    year = "2020"
}

@article{Adem,
    author = "Adem, Alejandro and Maginnis, John and Milgram, R. James",
    title = "{Symmetric invariants and cohomology of groups}",
    eprint = "2501.18399",
    archivePrefix = "arXiv",
    primaryClass = "math-ph",
    doi = "10.1007/BF01446902",
    journal = "Mathematische Annalen",
    volume = "287",
    number = "1",
    pages = "391--411",
    year = "1990"
}

@article{ABGHR14a,
    AUTHOR = {Ando, Matthew and Blumberg, Andrew J. and Gepner, David and
	      Hopkins, Michael J. and Rezk, Charles},
     TITLE = {An {$\infty$}-categorical approach to {$R$}-line bundles,
	      {$R$}-module {T}hom spectra, and twisted {$R$}-homology},
   JOURNAL = {J. Topol.},
  FJOURNAL = {Journal of Topology},
    VOLUME = {7},
      YEAR = {2014},
    NUMBER = {3},
     PAGES = {869--893},
      ISSN = {1753-8416},
   MRCLASS = {55P43 (55N20 55U40)},
  MRNUMBER = {3252967},
MRREVIEWER = {Tyler D. Lawson},
       DOI = {10.1112/jtopol/jtt035},
	note = "\url{arXiv:1403.4325}"
}

@article{ABGHR14b,
    AUTHOR = {Ando, Matthew and Blumberg, Andrew J. and Gepner, David and
	      Hopkins, Michael J. and Rezk, Charles},
     TITLE = {Units of ring spectra, orientations and {T}hom spectra via
	      rigid infinite loop space theory},
   JOURNAL = {J. Topol.},
  FJOURNAL = {Journal of Topology},
    VOLUME = {7},
      YEAR = {2014},
    NUMBER = {4},
     PAGES = {1077--1117},
      ISSN = {1753-8416},
   MRCLASS = {55P43 (55N34 55R65)},
  MRNUMBER = {3286898},
MRREVIEWER = {Steven R. Costenoble},
       DOI = {10.1112/jtopol/jtu009},
	note = "\url{arXiv:1403.4320}"
}

@article{Barkeshli:2021ypb,
    author = "Barkeshli, Maissam and Chen, Yu-An and Hsin, Po-Shen and Manjunath, Naren",
    title = "{Classification of (2+1)D invertible fermionic topological phases with symmetry}",
    eprint = "2109.11039",
    archivePrefix = "arXiv",
    primaryClass = "cond-mat.str-el",
    doi = "10.1103/PhysRevB.105.235143",
    journal = "Phys. Rev. B",
    volume = "105",
    number = "23",
    pages = "235143",
    year = "2022"
}

@article{Barkeshli:2022edm,
    author = "Barkeshli, Maissam and Chen, Yu-An and Hsin, Po-Shen and Kobayashi, Ryohei",
    title = "{Higher-group symmetry in finite gauge theory and stabilizer codes}",
    eprint = "2211.11764",
    archivePrefix = "arXiv",
    primaryClass = "cond-mat.str-el",
    doi = "10.21468/SciPostPhys.16.4.089",
    journal = "SciPost Phys.",
    volume = "16",
    number = "4",
    pages = "089",
    year = "2024"
}

@article{Barkeshli:2023bta,
    author = "Barkeshli, Maissam and Hsin, Po-Shen and Kobayashi, Ryohei",
    title = "{Higher-group symmetry of (3+1)D fermionic $\mathbb{Z}_2$ gauge theory: Logical CCZ, CS, and T gates from higher symmetry}",
    eprint = "2311.05674",
    archivePrefix = "arXiv",
    primaryClass = "cond-mat.str-el",
    doi = "10.21468/SciPostPhys.16.5.122",
    journal = "SciPost Phys.",
    volume = "16",
    number = "5",
    pages = "122",
    year = "2024"
}

@article{Ben88,
    AUTHOR = {Benson, Dave},
     TITLE = {Spin modules for symmetric groups},
   JOURNAL = {J. London Math. Soc. (2)},
  FJOURNAL = {Journal of the London Mathematical Society. Second Series},
    VOLUME = {38},
      YEAR = {1988},
    NUMBER = {2},
     PAGES = {250--262},
      ISSN = {0024-6107},
   MRCLASS = {20C30},
  MRNUMBER = {966297},
MRREVIEWER = {A. Kerber},
       DOI = {10.1112/jlms/s2-38.2.250},
       URL = {https://doi.org/10.1112/jlms/s2-38.2.250},
}

@article{BBSNT,
    author = "Bhardwaj, Lakshya and Bottini, Lea E. and Schafer-Nameki, Sakura and Tiwari, Apoorv",
    title = "{Non-invertible symmetry webs}",
    eprint = "2212.06842",
    archivePrefix = "arXiv",
    primaryClass = "hep-th",
    doi = "10.21468/SciPostPhys.15.4.160",
    journal = "SciPost Phys.",
    volume = "15",
    number = "4",
    pages = "160",
    year = "2023"
}

@article{BSNW,
    author = "Bhardwaj, Lakshya and Schafer-Nameki, Sakura and Wu, Jingxiang",
    title = "{Universal Non-Invertible Symmetries}",
    eprint = "2208.05973",
    archivePrefix = "arXiv",
    primaryClass = "hep-th",
    doi = "10.1002/prop.202200143",
    journal = "Fortsch. Phys.",
    volume = "70",
    number = "11",
    pages = "2200143",
    year = "2022"
}

@article{Bhardwaj:2024xcx,
    author = "Bhardwaj, Lakshya and D\'ecoppet, Thibault and Schafer-Nameki, Sakura and Yu, Matthew",
    title = "{Fusion 3-Categories for Duality Defects}",
    eprint = "2408.13302",
    archivePrefix = "arXiv",
    primaryClass = "math.CT",
    year = "2024",
    note = "\url{arXiv:2408.13302}"
}

@article{Bhardwaj:2024qiv,
    author = "Bhardwaj, Lakshya and Pajer, Daniel and Schafer-Nameki, Sakura and Tiwari, Apoorv and Warman, Alison and Wu, Jingxiang",
    title = "{Gapped Phases in (2+1)d with Non-Invertible Symmetries: Part I}",
    eprint = "2408.05266",
    archivePrefix = "arXiv",
    primaryClass = "hep-th",
    year = "2024",
   note = "\url{arXiv:2408.05266}"
}

@article{Bhardwaj:2025piv,
    author = "Bhardwaj, Lakshya and Schafer-Nameki, Sakura and Tiwari, Apoorv and Warman, Alison",
    title = "{Gapped Phases in (2+1)d with Non-Invertible Symmetries: Part II}",
    eprint = "2502.20440",
    archivePrefix = "arXiv",
    primaryClass = "hep-th",
    month = "2",
    year = "2025",
    note = "\url{arXiv:2502.20440}"
}

@article{Bhardwaj:2025jtf,
    author = "Bhardwaj, Lakshya and Gai, Yuhan and Huang, Sheng-Jie and Inamura, Kansei and Schafer-Nameki, Sakura and Tiwari, Apoorv and Warman, Alison",
    title = "{Gapless Phases in (2+1)d with Non-Invertible Symmetries}",
    eprint = "2503.12699",
    archivePrefix = "arXiv",
    primaryClass = "cond-mat.str-el",
    year = "2025",
   note="\url{arXiv:2503.12699}"
}

@article{BBFP:I,
    author =        "Thomas Bartsch and Mathew Bullimore and Andrea E. V. Ferrari and Jamie Pearson",
    title =         "Non-invertible Symmetries and Higher Representation Theory {I}",
    journal =       "SciPost Phys.",
    volume =        "17",
    number =        "1",
    pages =         "015",
    year =          "2024",
    note =          "\url{arXiv:2208.05993}",
}

@article{BBFP,
    author =        "Thomas Bartsch and Mathew Bullimore and Andrea E. V. Ferrari and Jamie Pearson",
    title =         "Non-invertible Symmetries and Higher Representation Theory {II}",
    journal =       "SciPost Phys.",
    volume =        "17",
    number =        "2",
    pages =         "067",
    year =          "2024",
    note =          "\url{arXiv:2212.07393}",
}

@article{bruillard2017fermionic,
  title={Fermionic modular categories and the 16-fold way},
  author={Bruillard, Paul and Galindo, C{\'e}sar and Hagge, Tobias and Ng, Siu-Hung and Plavnik, Julia Yael and Rowell, Eric C and Wang, Zhenghan},
  journal={Journal of Mathematical Physics},
  volume={58},
  number={4},
  year={2017},
  publisher={AIP Publishing}
}

@article{davydov2013witt,
  title={The {W}itt group of non-degenerate braided fusion categories},
  author={Davydov, Alexei and M{\"u}ger, Michael and Nikshych, Dmitri and Ostrik, Victor},
  journal={Journal f{\"u}r die reine und angewandte Mathematik (Crelles Journal)},
  volume={2013},
  number={677},
  pages={135--177},
  year={2013},
  publisher={De Gruyter}
}

@article{Debray:2022wcd,
    author = "Debray, Arun and Yu, Matthew",
    title = "{What Bordism-Theoretic Anomaly Cancellation Can Do for U}",
    eprint = "2210.04911",
    archivePrefix = "arXiv",
    primaryClass = "hep-th",
    doi = "10.1007/s00220-024-04937-4",
    journal = "Commun. Math. Phys.",
    volume = "405",
    number = "7",
    pages = "154",
    year = "2024"
}

@article{Debray:2023tdd,
    author = "Debray, Arun and Yu, Matthew",
    title = "{Adams spectral sequences for non-vector-bundle Thom spectra}",
    eprint = "2305.01678",
    archivePrefix = "arXiv",
    primaryClass = "math.AT",
    year = "2023",
   note= "\url{arXiv:2305.01678}"
}

@article{Debray:2023iwf,
    author = "Debray, Arun and Ye, Weicheng and Yu, Matthew",
    title = "{Bosonization and Anomaly Indicators of (2+1)-D Fermionic Topological Orders}",
    eprint = "2312.13341",
    archivePrefix = "arXiv",
    primaryClass = "math-ph",
    year = "2023",
    note = "\url{arXiv:2312.13341}"
}

@article{DYY1,
    author = "Debray, Arun and Ye, Weicheng and Yu, Matthew",
    title = "{How to Build Anomalous (3+1)d Topological Quantum Field Theories}",
    eprint = "2510.24834",
    archivePrefix = "arXiv",
    primaryClass = "math-ph",
    month = "10",
    year = "2025"
}

@article{decoppet2022drinfeld,
  title={Drinfeld centers and {M}orita equivalence classes of fusion 2-categories},
  author={D{\'e}coppet, Thibault D},
  journal={Compositio Mathematica},
  volume={161},
  number={2},
  pages={305--340},
  year={2025},
  publisher={London Mathematical Society}
}

@article{Decoppet:2022dnz,
    author = "D\'ecoppet, Thibault D. and Yu, Matthew",
    title = "{Gauging noninvertible defects: a 2-categorical perspective}",
    eprint = "2211.08436",
    archivePrefix = "arXiv",
    primaryClass = "math.CT",
    doi = "10.1007/s11005-023-01655-1",
    journal = "Lett. Math. Phys.",
    volume = "113",
    number = "2",
    pages = "36",
    year = "2023"
}

@article{decoppet2024extension,
  title={Extension Theory and Fermionic Strongly Fusion 2-Categories (with an Appendix by {T}hibault {D}idier {D}{\'e}coppet and {T}heo {J}ohnson-{F}reyd)},
  author={{D}{\'e}coppet, {T}hibault {D}idier},
  journal={SIGMA. Symmetry, Integrability and Geometry: Methods and Applications},
  volume={20},
  pages={092},
  year={2024},
  publisher={SIGMA. Symmetry, Integrability and Geometry: Methods and Applications}
}

@article{Decoppet:2024htz,
    author = "D\'ecoppet, Thibault D. and Huston, Peter and Johnson-Freyd, Theo and Nikshych, Dmitri and Penneys, David and Plavnik, Julia and Reutter, David and Yu, Matthew",
    title = "{The Classification of Fusion 2-Categories}",
    eprint = "2411.05907",
    archivePrefix = "arXiv",
    primaryClass = "math.CT",
    month = "11",
    year = "2024",
    note= "\url{arXiv:2411.05907}"
}

@article{Decoppet:2025eic,
    author = "D{\'e}coppet, Thibault D. and Yu, Matthew",
    title = "{The Classification of 3+1d Symmetry Enriched Topological Order}",
    eprint = "2509.10603",
    archivePrefix = "arXiv",
    primaryClass = "math-ph",
    month = "9",
    year = "2025"
}

@article{DelT,
    author =        "Clement Delcamp and Apoorv Tiwari",
    title =         "Higher categorical symmetries and gauging in two-dimensional spin systems",
    journal =       "SciPost Phys.",
    volume =        "16",
    number =        "4",
    pages =         "110",
    year =          "2024",
    note =          "\url{arXiv:2301.01259}",
}

@article{douglas2018fusion,
  title={Fusion 2-categories and a state-sum invariant for 4-manifolds},
  author={Douglas, Christopher L and Reutter, David J},
  year={2018},
  note= "\url{arXiv:1812.11933}"
}

@article{Ellison:2024svg,
    author = "Ellison, Tyler D. and Cheng, Meng",
    title = "{Toward a Classification of Mixed-State Topological Orders in Two Dimensions}",
    eprint = "2405.02390",
    archivePrefix = "arXiv",
    primaryClass = "cond-mat.str-el",
    doi = "10.1103/PRXQuantum.6.010315",
    journal = "PRX Quantum",
    volume = "6",
    number = "1",
    pages = "010315",
    year = "2025"
}

@article{etingof2010fusion,
  title={Fusion categories and homotopy theory},
  author={Etingof, Pavel and Nikshych, Dmitri and Ostrik, Victor},
  journal={Quantum topology},
  volume={1},
  number={3},
  pages={209--273},
  year={2010}
}

@article{Freed:2006mx,
    author = "Freed, Daniel S.",
    title = "{Pions and Generalized Cohomology}",
    eprint = "hep-th/0607134",
    archivePrefix = "arXiv",
    journal = "J. Diff. Geom.",
    volume = "80",
    number = "1",
    pages = "45--77",
    year = "2008"
}

@article{Gaiotto:2017zba,
    author = "Gaiotto, Davide and Johnson-Freyd, Theo",
    title = "{Symmetry Protected Topological phases and Generalized Cohomology}",
    eprint = "1712.07950",
    archivePrefix = "arXiv",
    primaryClass = "hep-th",
    doi = "10.1007/JHEP05(2019)007",
    journal = "JHEP",
    volume = "05",
    pages = "007",
    year = "2019"
}

@article{Gaiotto:2019xmp,
    author = "Gaiotto, Davide and Johnson-Freyd, Theo",
    title = "{Condensations in higher categories}",
    eprint = "1905.09566",
    archivePrefix = "arXiv",
    primaryClass = "math.CT",
    year = "2019",
    note= "\url{arXiv:1905.09566}"
}

@article{galindo2017categorical,
  title={Categorical fermionic actions and minimal modular extensions},
  author={Galindo, C{\'e}sar and Venegas-Ramirez, C{\'e}sar F},
  year={2017},
  note = "\url{arXiv:1712.07097}"
}

@article{Gu:2012ib,
    author = "Gu, Zheng-Cheng and Wen, Xiao-Gang",
    title = "{Symmetry-protected topological orders for interacting fermions: Fermionic topological nonlinear \ensuremath{\sigma} models and a special group supercohomology theory}",
    eprint = "1201.2648",
    archivePrefix = "arXiv",
    primaryClass = "cond-mat.str-el",
    doi = "10.1103/PhysRevB.90.115141",
    journal = "Phys. Rev. B",
    volume = "90",
    number = "11",
    pages = "115141",
    year = "2014"
}

@article{Haah:2011drr,
    author = "Haah, Jeongwan",
    title = "{Local stabilizer codes in three dimensions without string logical operators}",
    eprint = "1101.1962",
    archivePrefix = "arXiv",
    primaryClass = "quant-ph",
    reportNumber = "CALT-68-2816, CALT-68-2816",
    doi = "10.1103/physreva.83.042330",
    journal = "Phys. Rev. A",
    volume = "83",
    number = "4",
    pages = "042330",
    year = "2011"
}

@article{haah2014bifurcation,
  title={Bifurcation in entanglement renormalization group flow of a gapped spin model},
  author={Haah, Jeongwan},
  journal={Physical Review B},
  volume={89},
  number={7},
  pages={075119},
  year={2014},
  publisher={APS}
}

@article{Hsin:2019gvb,
    author = "Hsin, Po-Shen and Shao, Shu-Heng",
    title = "{Lorentz Symmetry Fractionalization and Dualities in (2+1)d}",
    eprint = "1909.07383",
    archivePrefix = "arXiv",
    primaryClass = "cond-mat.str-el",
    reportNumber = "CALT-TH-2019-035",
    doi = "10.21468/SciPostPhys.8.2.018",
    journal = "SciPost Phys.",
    volume = "8",
    pages = "018",
    year = "2020"
}

@article{JFS,
  title={(Op) lax natural transformations, twisted quantum field theories, and “even higher” Morita categories},
  author={Johnson-Freyd, Theo and Scheimbauer, Claudia},
  journal={Advances in Mathematics},
  volume={307},
  pages={147--223},
  year={2017},
  publisher={Elsevier}
}

@article{JF,
    title =         "On the classification of topological orders",
    author =        "Theo Johnson-Freyd",
    journal=        {Communications in Mathematical Physics},
    volume=         {393},
    number=         {2},
    pages=          {989-1033},
    year =          {2022},
    note =          "\url{arXiv:2003.06663}",
}

@article{Johnson-Freyd:2020ivj,
    author = "Johnson-Freyd, Theo and Yu, Matthew",
    title = "{Fusion 2-categories With no Line Operators are Grouplike}",
    eprint = "2010.07950",
    archivePrefix = "arXiv",
    primaryClass = "math.QA",
    doi = "10.1017/S0004972721000095",
    journal = "Bull. Austral. Math. Soc.",
    volume = "104",
    number = "3",
    pages = "434--442",
    year = "2021"
}

@article{Johnson-Freyd:2020twl,
    author = "Johnson-Freyd, Theo",
    title = "{(3+1)D topological orders with only a $\mathbb{Z}_2$-charged particle}",
    eprint = "2011.11165",
    archivePrefix = "arXiv",
    primaryClass = "math.QA",
    year = "2020",
    note = "\url{arXiv:2011.11165}"
}

@article{Johnson-Freyd:2021tbq,
    author = "Johnson-Freyd, Theo and Yu, Matthew",
    title = "{Topological Orders in (4+1)-Dimensions}",
    eprint = "2104.04534",
    archivePrefix = "arXiv",
    primaryClass = "hep-th",
    doi = "10.21468/SciPostPhys.13.3.068",
    journal = "SciPost Phys.",
    volume = "13",
    number = "3",
    pages = "068",
    year = "2022"
}

@article{johnson2024minimal,
  title={Minimal nondegenerate extensions},
  author={Johnson-Freyd, Theo and Reutter, David},
  journal={Journal of the American Mathematical Society},
  volume={37},
  number={1},
  pages={81--150},
  year={2024}
}

@misc{Johnson-Freyd:2022,
    author = "Johnson-Freyd, Theo",
    title = "{A 4D TQFT that is not quite a gauge theory}",
    year = "2022",
 note = "\url{http://categorified.net/Nongauge-SymSem.pdf}"
}

@article{Kapustin:2017jrc,
    author = "Kapustin, Anton and Thorngren, Ryan",
    title = "{Fermionic {SPT} phases in higher dimensions and bosonization}",
    eprint = "1701.08264",
    archivePrefix = "arXiv",
    primaryClass = "cond-mat.str-el",
    doi = "10.1007/JHEP10(2017)080",
    journal = "JHEP",
    volume = "10",
    pages = "080",
    year = "2017"
}

@article{Kitaev:2005hzj,
    author = "Kitaev, Alexei",
    title = "{Anyons in an exactly solved model and beyond}",
    eprint = "cond-mat/0506438",
    archivePrefix = "arXiv",
    doi = "10.1016/j.aop.2005.10.005",
    journal = "Annals Phys.",
    volume = "321",
    number = "1",
    pages = "2--111",
    year = "2006"
}

@article{Kong:2014qka,
    author = "Kong, Liang and Wen, Xiao-Gang",
    title = "{Braided fusion categories, gravitational anomalies, and the mathematical framework for topological orders in any dimensions}",
    eprint = "1405.5858",
    archivePrefix = "arXiv",
    primaryClass = "cond-mat.str-el",
    year = "2014",
    note = "\url{arXiv:1405.5858}"
}

@article{kong2015boundary,
  title={Boundary-bulk relation for topological orders as the functor mapping higher categories to their centers},
  author={Kong, Liang and Wen, Xiao-Gang and Zheng, Hao},
  note={\url{arXiv:1502.01690}},
  year={2015}
}

@article{kong2017boundary,
  title={Boundary-bulk relation in topological orders},
  author={Kong, Liang and Wen, Xiao-Gang and Zheng, Hao},
  journal={Nuclear Physics B},
  volume={922},
  pages={62--76},
  year={2017},
  publisher={Elsevier}
}

@article{Kong:2020jne,
    author = "Kong, Liang and Lan, Tian and Wen, Xiao-Gang and Zhang, Zhi-Hao and Zheng, Hao",
    title = "{Classification of topological phases with finite internal symmetries in all dimensions}",
    eprint = "2003.08898",
    archivePrefix = "arXiv",
    primaryClass = "math-ph",
    doi = "10.1007/JHEP09(2020)093",
    journal = "JHEP",
    volume = "09",
    pages = "093",
    year = "2020"
}

@article{Lan_2018,
   title={Classification of 
(3+1)D
 Bosonic Topological Orders: The Case When Pointlike Excitations Are All Bosons},
   volume={8},
   ISSN={2160-3308},
   url={http://dx.doi.org/10.1103/PhysRevX.8.021074},
   DOI={10.1103/physrevx.8.021074},
   number={2},
   journal={Physical Review X},
   publisher={American Physical Society (APS)},
   author={Lan, Tian and Kong, Liang and Wen, Xiao-Gang},
   year={2018},
   month={Jun}
}

@article{Lan_2019,
   title={Classification of 
3+1D
 Bosonic Topological Orders (II): The Case When Some Pointlike Excitations Are Fermions},
   volume={9},
   ISSN={2160-3308},
   url={http://dx.doi.org/10.1103/PhysRevX.9.021005},
   DOI={10.1103/physrevx.9.021005},
   number={2},
   journal={Physical Review X},
   publisher={American Physical Society (APS)},
   author={Lan, Tian and Wen, Xiao-Gang},
   year={2019},
   month={Apr}
}

@article{lan2017modular,
  title={Modular extensions of unitary braided fusion categories and 2+ 1D topological/{SPT} orders with symmetries},
  author={Lan, Tian and Kong, Liang and Wen, Xiao-Gang},
  journal={Communications in Mathematical Physics},
  volume={351},
  pages={709--739},
  year={2017},
  publisher={Springer}
}

@article{nikshych2022computing,
  title={Computing the group of minimal non-degenerate extensions of a super-{T}annakian category},
  author={Nikshych, Dmitri},
  journal={Communications in Mathematical Physics},
  volume={396},
  number={2},
  pages={685--711},
  year={2022},
  publisher={Springer}
}

@article{Sohal:2024qvq,
    author = "Sohal, Ramanjit and Prem, Abhinav",
    title = "{Noisy Approach to Intrinsically Mixed-State Topological Order}",
    eprint = "2403.13879",
    archivePrefix = "arXiv",
    primaryClass = "cond-mat.str-el",
    doi = "10.1103/PRXQuantum.6.010313",
    journal = "PRX Quantum",
    volume = "6",
    number = "1",
    pages = "010313",
    year = "2025"
}

@article{Tachikawa:2017gyf,
    author = "Tachikawa, Yuji",
    title = "{On gauging finite subgroups}",
    eprint = "1712.09542",
    archivePrefix = "arXiv",
    primaryClass = "hep-th",
    reportNumber = "IPMU-17-0183",
    doi = "10.21468/SciPostPhys.8.1.015",
    journal = "SciPost Phys.",
    volume = "8",
    number = "1",
    pages = "015",
    year = "2020"
}

@article{Yu:2020twi,
    author = "Yu, Matthew",
    title = "{Symmetries and anomalies of (1+1)d theories: 2-groups and symmetry fractionalization}",
    eprint = "2010.01136",
    archivePrefix = "arXiv",
    primaryClass = "hep-th",
    doi = "10.1007/JHEP08(2021)061",
    journal = "JHEP",
    volume = "08",
    pages = "061",
    year = "2021"
}

@article{Yu:2021zmu,
    author = "Yu, Matthew",
    title = "{Gauging Categorical Symmetries in 3d Topological Orders and Bulk Reconstruction}",
    eprint = "2111.13697",
    archivePrefix = "arXiv",
    primaryClass = "hep-th",
    month = "11",
    year = "2021"
}

@article{Zhang:2023upl,
    author = "Zhang, Carolyn",
    title = "{Note on quantum cellular automata and strong equivalence}",
    eprint = "2306.03171",
    archivePrefix = "arXiv",
    primaryClass = "quant-ph",
    year = "2023",
    note = "\url{ arXiv:2306.03171v1}"
}
\end{document}